\documentclass[12pt]{article}

\usepackage{latexsym,amsmath,amscd,amssymb,graphics,cite}
\usepackage{enumerate}
\usepackage{mathtools}
\usepackage{graphicx}
\usepackage{framed}
\usepackage{url}
\usepackage{cite}
\usepackage{cancel}

\usepackage{authblk}

\usepackage[all]{xy}

\makeatletter

\@addtoreset{figure}{section}
\def\thefigure{\thesection.\@arabic\c@figure}
\def\fps@figure{h, t}
\@addtoreset{table}{bsection}
\def\thetable{\thesection.\@arabic\c@table}
\def\fps@table{h, t}
\@addtoreset{equation}{section}

\makeatother

\newtheorem{theorem}{Theorem}

\newtheorem{lemma}[theorem]{Lemma}

\newtheorem{remark}[theorem]{Remark}

\numberwithin{theorem}{section}
\newenvironment{proof}[1][Proof]{\textbf{#1.} }{\ \rule{0.5em}{0.5em}}

\def\be{\begin{equation}}
\def\ee{\end{equation}}
\def\bea{\begin{eqnarray}}
\def\eea{\end{eqnarray}}
\def\ba{\begin{array}}
\def\ea{\end{array}}

\def\bOm{\boldsymbol{\Omega}}

\let\<\langle
\let \>\rangle

\newcommand{\rem}[1]{}

\newcommand{\de}{\delta}

\newcommand{\bv}{\boldsymbol{v}}

\newcommand{\bGam}{\boldsymbol{\Gamma}}
\newcommand{\bom}{\boldsymbol{\omega}}
\newcommand{\bgam}{\boldsymbol{\gamma}}
\newcommand{\bsigma}{\boldsymbol{\Sigma}}
\newcommand{\bpsi}{\boldsymbol{\Psi}}

\newcommand{\bchi}{\boldsymbol{\chi}} 

\newcommand{\pp}[2]{\frac{\partial #1}{\partial #2}}

\newcommand{\dede}[2]{\frac{\delta #1}{\delta #2}}
\newcommand{\prt}{\partial}

\newcommand{\lp}{\left(}
\newcommand{\rp}{\right)}

\newcommand{\mse}{\mathfrak{se}}

\begin{document}

\markboth{R. Ivanov and V. Putkaradze}{Swirling fluid flow in elastic tubes}

\title{Swirling fluid flow in flexible, expandable elastic tubes: variational approach, reductions and integrability}
\author[1]{Rossen Ivanov} 

\affil[1]{School of Mathematical Sciences, Technological University Dublin, City Campus, Kevin Street,
Dublin D08 NF82,
Ireland; 
rossen.ivanov@dit.ie }

\author[2]{Vakhtang Putkaradze$^*$ } 

\affil[2]{Department of Mathematical and Statistical Sciences
632 CAB University of Alberta
Edmonton, AB,  T6G 2G1 Canada;  
putkarad@ualberta.ca; $^*$(corresponding author). 
} 


\maketitle

\begin{abstract}
Many engineering and physiological applications deal with situations when a fluid is moving in flexible tubes with elastic walls. In real-life applications like blood flow, a swirl in the fluid often plays an important role, presenting an additional complexity not described by previous theoretical models. We present a theory for the dynamics of the interaction between elastic tubes and swirling fluid flow.  The equations are derived using a variational principle, with the incompressibility constraint of the fluid giving rise to a pressure-like term. In order to connect this work with the previous literature, we consider the case of inextensible and unshearable tube with a straight centerline. In the absence of vorticity, our model reduces to previous models considered in the literature, yielding the equations of conservation of fluid momentum, wall momentum and the fluid volume. We pay special attention to the case when the vorticity is present but kept at a constant value. We show the conservation of energy-like quality and find an additional momentum-like conserved quantity. Next, we develop an alternative formulation, reducing the system of three conservation equations to a single compact equation for the back-to-labels map. That single equation shows interesting instability in solutions when the velocity exceeds a critical value. Furthermore, the equation in stable regime can be reduced to Boussinesq-type, KdV and Monge-Amp\`ere equations in several appropriate limits, namely, the first two in the limit of a long time and length scales and the third one in the additional limit of the small cross-sectional area. For the unstable regime, the numerical solutions demonstrate the spontaneous appearance of large oscillations in the cross-sectional area. 
\\[1mm] 
\\[1mm] 
\emph{Declarations of interest}: none. 
\end{abstract}
\tableofcontents

\section{Introduction}
Models of flexible tubes conveying fluid have been under investigation for many decades, because of their importance to industrial and biomedical applications. Broadly speaking, the research can be put in two different categories. The first category of papers focused on the instability of centerline of the tubes due to the complex interaction of the fluid inside the tube and the tube's elasticity, most often treated as an elastic rod. For such systems, an instability appears when the flow rate through the tube exceeds a certain critical value, sometimes referred to as the \emph{garden hose instability}.  While this phenomenon has been known for a very long time, the quantitative research in the field started around 1950 \cite{AsHa1950}. We believe that Benjamin \cite{Be1961a,Be1961b} was the first to formulate a quantitative theory for the 2D dynamics  of the {\it initially straight tubes} by considering a linked chain of tubes conveying fluids and using an augmented Hamilton principle of critical action that takes into account the momentum of the jet leaving the tube.  A continuum equation for the linear disturbances was then derived as the limit of the discrete system. Interestingly, Benjamin's approach is the inverse to the derivation of variational integrators for continuum equations \cite{Ve1988,MoVe1991,MaWe2001} This linearized equation for the initially straight tubes was further considered by Gregory and Pa\"idoussis \cite{GrPa1966a}. 

These initial developments formed the basis for further stability analysis of this problem for finite, initially straight tubes \cite{Pa1970,PaIs1974,Pa1998,ShMi2001,DoLa2002,PaLi1993,Pa2004,AkIvKoNe2013,AkGeNe2015,AkGeNe2016}. The linear stability theory has shown a reasonable agreement with experimentally observed onset of the instability \cite{GrPa1966b,Pa1998,KuSa2005,CaCr2009,Cr-etal-2012}. Nonlinear deflection models  were also considered in \cite{SeLiPa1994,Pa2004,MoPa2009,GaPaAm2013}, and the compressible (acoustic) effects in the flowing fluid in \cite{Zh2008}. Alternatively,  a more detailed 3D theory of motion was developed in \cite{BeGoTa2010} and extended in \cite{RiPe2015}, based on a modification of the Cosserat rod treatment for the description of elastic dynamics of the tube, while keeping the cross-section of the tube constant and orthogonal to the centerline. In particular, \cite{RiPe2015} analyzes several non-straight configurations, such as tube hanging under the influence of gravity,  both from the point of view of linear stability and nonlinear behavior. Unfortunately, this Cosserat-based theory could not easily incorporate the effects of the cross-sectional changes in the dynamics. Some authors have treated the instability from the point of view of the {\it follower force approach}, which treats the system as an elastic beam, ignoring the fluid motion,  with a force that is always tangent to the end of the tube. Such a force models the effect of the jet leaving the nozzle \cite{BoRoSa2002}, showing interesting dynamics of solutions. However, once the length of the tube becomes large, the validity of the follower force approach has been questioned, see  \cite{El2005} for a lively and thorough discussion.
For the history of the development of this problem in the Soviet/Russian literature, we refer the reader to the monograph \cite{Sv1987} (still only available in Russian). To briefly touch upon the developments in Russian literature that have been published in parallel with their western counterparts, and perhaps less known in the west, we refer the reader to the selection of papers, with most of them now available in English  \cite{Mo1965,Mu1965,Il1969,Anni1970,VoGr1973,Sv1978,Do1979,Ch1984,So2005,AmAl2015}. 

It is also worth noting the developments in the theory of initially curved pipes conveying fluid, which has been considered in  in some detail in the literature. The equations of motion for such theory were initially derived using the balance of elastic forces from tube's deformation and  fluid forces acting on the tube when the fluid is moving along a curved line in space.
In the western literature, we shall mention the earlier work \cite{Ch1972}, followed with more detailed studies \cite{MiPaVa1988a,MiPaVa1988b,DuRo1992} which developed the theory suited for both extensible and inextensible tubes and discussed the finite-element method realization of the problem. We shall also mention \cite{DoMo1976,AiGi1990} deriving a variational approach for the planar motions of initially circular tubes, although the effect of curved fluid motion was still introduced as extra forces through the Lagrange-d'Alembert principle.  In the Soviet/Russian literature,   \cite{Sv1978} developed the rod-based theory of oscillations and \cite{So2005} considered an improved treatment of forces acting on the tubes. Most of the work has been geared towards the understanding of the planar cases with in-plane vibrations as the simplest and most practically relevant situations (still, however, leading to quite complex formulas). 

In spite of substantial amount of literature in previous work in the area, many questions remain difficult to answer using the traditional approach.  Most importantly, it is very difficult (and perhaps impossible)  to extend the previous theory to accurately take into account the changes in the cross-sectional area of the tube, also called the collapsible tube case. In many previous works, the effects of cross-sectional changes have  been considered through the quasi-static approximation: if $A(s,t)$ is the local cross-section area, and $u(s,t)$ is the local velocity of the fluid, with $s$ being the coordinate along the tube and $t$ the time,  then the quasi-static assumption states that  $uA=$ const, see, \emph{e.g.}, \cite{SeLiPa1994,GhPaAm2013}. Unfortunately, this simple law is not correct in general and should  only be used for steady flows. This problem has been solved in the fully variational derivation \cite{FGBPu2014,FGBPu2015}, where a geometrically exact setting for dealing with a variable cross-section depending on tube's deformation was developed and studied, showing the important effects of the cross-sectional changes on both linear and nonlinear dynamics. The nonlinear theory was derived from a variational principle in a rigorous geometric setting and for general Lagrangians. It can incorporate general boundary conditions and arbitrary deviations from equilibrium in the three-dimensional space.  From a mathematical point of view, the Lagrangian description of these systems involves both left-invariant (elastic) and right-invariant (fluid) quantities.  This method has since been applied to the linear stability of initially curved (helical) tubes in \cite{FGBGePu2018}, where it was shown that the geometric method of 
\cite{FGBPu2014,FGBPu2015} is guaranteed to yield an equation with constant coefficients for the linear stability analysis for any helical initial state. 
This theory further allowed consistent variational approximations of the solutions, both from the point of view of deriving simplified reduced models and developing structure preserving numerical schemes \cite{FGBPu2016}, echoing the original development of the equations by Benjamin \cite{Be1961a,Be1961b} for constant cross-section. 

The second category of papers concerns with the flow of fluid through tubes which can dynamically change its cross-section. In these works, the cross-section, or its radius, in the case of circular cross-sections, serves as an additional dynamic variable. The equations of motion then involve the balance of fluid momentum, wall momentum and conservation law for the fluid. The theory for such tubes is indispensable in biomedical applications, such as arterial flows, and the key initial progress in theoretical understanding is based on balancing the conservation laws in a collapsible tube, see \cite{Pe1980}.  Further work in this field explored applications to arterial \cite{Je1990,PeLu1998,LuPe1998,QuTuVe2000,FoLaQu2003,StWaJe2009,Tang-etal-2009} and lung flows \cite{ElKaSh1989,MaFl2010,Do2016}. Analytical studies for such flows are usually limited to cases when the centerline of the tube is straight. In addition, for analytical progress, one needs to further assume that material  particles on the wall can only move normally to the centerline during the dynamics (see, however, \cite{BuCaGlTaQu2013} for non-trivial longitudinal  displacement). We should also mention the  instability through the neck formation and self-sustaining flow pulsations suggested by \cite{Pe1992} and experimentally measured in \cite{KuMa1999}, the treatment of forward and backward running waves in arteries \cite{PaJo1990,JoPaHuSh1992} and the treatment of poroelastic walls \cite{Bu2016}.   While substantial progress in the analysis of the flow has been achieved so far, it was difficult to describe analytically the general dynamics of 3D deformations of the tube involving \emph{e.g.} the combination of shear, transversal deformation, and extensions. A lively discussion of the dichotomy between the needs of numerical solutions of Navier-Stokes equations in realistic geometries, and the importance of theoretical developments based on simplified models can be found in \cite{Pe2003}. 
 For more informations about the application of collapsible tubes to biological flows, we refer the reader to the reviews \cite{GrJe2004,HeHa2011,Se2016}. 
 
 A fully variational theory of motion of tubes with expandable walls in three dimensions was derived in \cite{FGBPu2019}, both for compressible and incompressible flow motion inside the tube. The major difference with the earlier models described by the variational approach \cite{FGBPu2014,FGBPu2015} was the ability to treat compressible flows, and incorporate the coupled dynamics of the walls, the tube and the fluid. This model can treat arbitrary motion of the tubes in three dimensions, using arbitrary elasticities and mass distribution of the walls. For compressible flows, the model is also capable to incorporate the motion of shock waves in a moving tubes with expandable walls. It was shown that all simplified analytical models of flows with expandable walls can be obtained as particular cases of the more generalized model, with the introduction of appropriate friction terms. As far as we are aware, there is no way to achieve such generality of motion with non-variational models. Indeed, one can guess the right terms in simplified geometries  using force balance, but trying to balance forces in more complex problems, such as deforming centerline and expandable walls, is fraught with possibilities for errors. In contrast, variational methods provide a systematic way of deriving the equations and their subsequent analysis, such as Hamiltonian formulation and the existence of Poisson brackets. We thus believe that the the problems in fluid-structure interactions can benefit greatly from the application of modern variational methods and geometric mechanics. 
 
Most of the analytical aspects of previous works discussed above have treated the flow of fluid in a simplified manner, treating one characteristic streamwise velocity for each point of the centerline. However,  the fluid flow regimes with vorticity are frequently encountered around stents in arteries, have been observed both experimentally and numerically. They are also known as 'spiral flows', or 'spiral laminar flows', and are believed to play a crucial role in arterial dynamics. The study of these flows with applications to blood flow dynamics started with the pioneering observation by Stonebridge \& Drophy \cite{StBr1991} and continued by many authors \cite{StHoAlBe1996spiral,MiSiUd2001,Gr-etal-2005,PaLa2009,BiHuGa2011,Re-etal-2014}.  Numerical studies of vorticity transport in complex geometries relevant to arteries can be found in \emph{e.g.} \cite{DoShFrPe2002}, but as far as we are aware, there has been no theoretical studies including the wall dynamics, fluid dynamics and vorticity, especially from the variational point of view. The goal of this paper is precisely to fill this theoretical gap. More precisely, we shall derive a fully three-dimensional flow of mean vorticity transport based on Kelvin's theorem of circulation in fluid mechanics. We will mostly focus on the variational derivation, and thus most of our considerations will be dedicated to the inviscid models. The  friction terms essential for blood flow applications can be added later as they have to be incorporated in the variational principle using the Lagrange-d'Alembert's method of external forces. Our model will be valid for arbitrary three-dimensional deformations of the centerline and arbitrary elasticity laws satisfied by the wall.  As we show, the simplified mathematical treatment for the flow can produce interesting mechanical effects and instability caused by the interaction of wall and fluid. As far as practical applications of the theory presented here, we view it as a relatively complete model of industrial flows in flexible and expandable tubes with high vorticity, and a first step towards understanding the complete structure of inertial terms in arterial models, to which terms containing the internal viscous friction and wall shear can later be added. 

A major role in this manuscript will be played by the back-to-label Lagrangian map for the fluid part, which maps the location of fluid particle at a given point on the tube at the time $t$ to its original location at $t=0$. The back-to-labels map has an important role in the development of continuum mechanics. Of particular value to our studies here is the Clebsch representation of back-to-labels map which was used to derive variational and Hamiltonian structure of continuum mechanics. One of the pioneering examples of this approach was developed in \cite{HoKu1983}, where Poisson brackets and Clebsch representation for fluid mechanics, magnetohydrodynamics (MHD), and elasticity theory were derived. Nonlinear stability theory of fluid flows using back-to-labels maps was investigated in 
\cite{Holm-etal-1985,AbHo1987}, and applications to quasi-geostrophic flows was studied in \cite{AlHo1996}. A multi-symplectic formulation of fluid mechanics for ideal fluids using back-to-labels maps was further derived in \cite{CoHoHy2007}. In these papers, variational derivatives with respect to back-to-labels map was crucial in obtaining the laws of motion and proving the Hamiltonian structure of the flow. Because of the features of 3D hydrodynamics, the expression of the incompressibility condition in terms of back-to-labels map is rather awkward, and thus the works cited here used variations with constrained Lagrangian and needed the pressure as the corresponding Lagrange multiplier. Our  problem  has the advantage that one can explicitly connect the cross-sectional area and the spatial derivatives of back-to-labels map, and substitute into the Lagrangian directly, leading to the unconstrained Euler-Lagrange equations.

\paragraph{Structure of the paper and main results} 
The paper is structured as follows. In Section~\ref{sec:general}, we derive the full three-dimensional equations of motion for flexible tubes carrying incompressible fluid which has streamwise velocity and and vorticity. We have tried to keep the derivation self-contained and pedagogical, and build up our theory from the  theory of exact geometric rods with no fluid, to the theory of fluid-conveying tubes with expandable walls, introducing the swirling motion of the fluid in the process. The fluid volume incompressibility condition leads to the appearance of a pressure-like term in the equations of motion. In Section~\ref{sec:1D_tubes}, we consider the case of the inextensible and unshearable tube, and a static straight centerline when the walls can only move normal to the centerline, which is the case often considered in the literature for theoretical studies.  The equations of motion then reduce to a generalization of the equations familiar from the previous literature, with four equations of motion describing the conservation of the fluid and wall momenta and the fluid's mass, and the transport of vorticity (or, more precisely, velocity circulation), with the corrections due to the vorticity terms. These equations explicitly contain the fluid pressure as the Lagrange multiplier for incompressibility. We show the conservation of energy for these equations and also derive a new constant of motion linear in momenta in \eqref{int_gen}, that seems to have not been observed in the previous works on the subject. By using an alternative derivation using the Lagrangian back-to-labels map, we show how to reduce this equation to a single nonlinear scalar equation of motion in Section~\ref{sec:single_eq}. In the linearized form, the equation shows interesting stability properties, in particular, instability for a sufficiently large velocity in the tube. We perform the numerical studies of the solution in both linearly stable and unstable regimes. For the linearly stable regime, in Section~\ref{sec:approximate_sols} we utilize the slow-time and large-wavelength approximation to first develop a reduction to a Monge-Amp\'ere's equation, then a Boussinesq-like equation and finally the KdV equation. We compare the numerics of the approximate models with the full equations and discuss the relevance of KdV approximations for the left- and right-running waves. 

\section{Exact geometric theory for flexible tubes conveying incompressible fluid}\label{sec:general}

In this section, we first quickly review the Lagrangian variational formulation for geometrically exact rods without fluid motion as the foundation of the theory. Then, we extend this variational formulation to incorporate the motion of the incompressible fluid inside the tube and the motion of the wall, and finish by introducing the vorticity. To achieve this goal we need to first identify the configuration manifold of the system (which is infinite-dimensional), as well as the convective and spatial variables for the tube and the fluid. 
Our model is then obtained by an application of the Hamilton principle, reformulated in convective variables for the tube and in spatial variables for the fluid. For rigorous justification of the variational approach to fluid-structure interactions see \cite{FGBPu2015,FGBPu2019}.

\subsection{Background for geometrically exact rod theory}

Here we briefly review the theory of geometrically exact rods following the approach developed, on the Hamiltonian side, in \cite{SiMaKr1988} and subsequently in the Lagrangian framework more appropriate to this article in \cite{HoPu2009,ElGBHoPuRa2010}. A more comprehensive introduction is contained in \cite{FGBPu2015,FGBGePu2018} to which we refer the reader for the details. The purely elastic (\emph{i.e.}, rods carrying no fluid) geometrically exact theory is equivalent to the Cosserat's rods \cite{CoCo1909}, the equivalence of these approaches was shown in \cite{SiMaKr1988,ElGBHoPuRa2010}.

The configuration of the rod deforming in the ambient space $ \mathbb{R}  ^3 $ is defined by specifying the position of its line of centroids by means of a map $\mathbf{r} (t,s)\in \mathbb{R}^3$, and by giving the orientation of the cross-section at that point. Here $t$ is the time and $s\in [0,L]$ is a parameter along the strand that does not need to be arclength. The orientation of the cross-section is given by a moving basis $\{ \mathbf{e} _i(t,s) \mid i=1,2,3\}$ attached to the cross section relative to a fixed frame $\{ \mathbf{E} _i \mid i=1,2,3\}$. The moving basis is described by an orthogonal transformation $ \Lambda (t,s) \in SO(3)$ such that $ \mathbf{e} _i (t,s)= \Lambda (t,s) \mathbf{E} _i $.
We interpret the maps $ \Lambda (t,s)$ and $\mathbf{r}(t,s)$ as a curve $t\mapsto (\Lambda(t),\mathbf{r}(t))\in G$ in the infinite dimensional Lie group $G= \mathcal{F} ([0,L], SO(3) \times \mathbb{R} ^3  )$ of $SO(3) \times \mathbb{R} ^3$-valued smooth maps defined on $[0,L]$  \footnote{The notation $ \mathcal{F} ([0,L], V)$ refers to the set of functions defined on the interval $s \in [0,L]$ taking the values  in the set $V$.}. The Lie group $G$ is the configuration manifold for the geometrically exact rod.

Following Hamilton's principle, given a Lagrangian function
\[
\mathsf{L}=\mathsf{L}( \Lambda , \dot {\Lambda }, \mathbf{r} , \dot{\mathbf{r}}):TG \rightarrow \mathbb{R},
\]
defined on the tangent bundle $TG$ of the configuration Lie group $G$ defined above, the equations of motion are the Euler-Lagrange equations obtained by the critical action principle
\begin{equation}\label{HP_rod} 
\delta \int_0^T\mathsf{L}( \Lambda , \dot {\Lambda }, \mathbf{r} , \dot{\mathbf{r}})\mbox{d}t=0,
\end{equation} 
for arbitrary variations $\delta\Lambda$ and $\delta  \mathbf{r} $ vanishing at $t=0,T$.
It turns out that the Lagrangian of geometrically exact rods can be exclusively expressed in terms of the convective variables 
\begin{equation}\label{def_conv_var}
\begin{aligned}
&\bgam=\Lambda^{-1} \dot{\mathbf{r} }\,,  &\quad 
&\bom=\Lambda^{-1} \dot\Lambda\,,
\\ 
&\bGam=\Lambda^{-1} \mathbf{r} '\,, &\quad
&\bOm=\Lambda^{-1} \Lambda'\,,
\end{aligned} 
\end{equation} 
see \cite{SiMaKr1988}, where $ \boldsymbol{\gamma} (t), \boldsymbol{\omega} (t)\in \mathcal{F} ([0,L], \mathbb{R}  ^3 )$ are the linear and angular convective velocities and $ \boldsymbol{\Gamma} (t), \boldsymbol{\Omega}(t)\in \mathcal{F} ([0,L], \mathbb{R}  ^3 )$ are the linear and angular convective strains. This gives rise to a Lagrangian $\ell=\ell( \boldsymbol{\omega} , \boldsymbol{\gamma} , \boldsymbol{\Omega} , \boldsymbol{\Gamma} ): \mathcal{F}([0,L], \mathbb{R}  ^3 )^4  \rightarrow \mathbb{R}  $ written exclusively in terms of convective variables. Here, we treat all four  variables as infinite-dimensional functions of $t$, which are themselves $\mathbb{R}^3$-valued functions of $s$ for any fixed $t$, see \cite{ElGBHoPuRa2010} for the detailed mathematical exposition of the method.  For the moment, we leave the Lagrangian function unspecified, we will give its explicit expression later in \S\ref{Def_Lagrangian} for the case of fluid-conveying tubes.

The equations of motion in convective description are obtained by writing the critical action principle \eqref{HP_rod} in terms of the Lagrangian $\ell$. This is accomplished by computing the constrained variations of $\boldsymbol{\omega} , \boldsymbol{\gamma} , \boldsymbol{\Omega} , \boldsymbol{\Gamma} $ induced by the free variations $ \delta \Lambda , \delta\mathbf{r} $ via the definitions \eqref{def_conv_var}. 
We find
\begin{align} 
& \delta \boldsymbol{\omega} = \frac{\partial \boldsymbol{\Sigma} }{\partial t} +\boldsymbol{\omega} \times \boldsymbol{\Sigma} , \qquad \delta \boldsymbol{\gamma} = \frac{\partial \boldsymbol{\Psi} }{\partial t} + \boldsymbol{\gamma} \times \boldsymbol{\Sigma} + \boldsymbol{\omega} \times \boldsymbol{\Psi} \, , 
\label{delta1} 
\\
& \delta \boldsymbol{\Omega} = \frac{\partial \boldsymbol{\Sigma} }{\partial s} +\boldsymbol{\Omega} \times \boldsymbol{\Sigma} , \qquad \delta \boldsymbol{\Gamma} = \frac{\partial \boldsymbol{\Psi} }{\partial s} + \boldsymbol{\Gamma} \times \boldsymbol{\Sigma} + \boldsymbol{\Omega} \times \boldsymbol{\Psi},
\label{delta2} 
\end{align} 
where $ \bsigma(t,s)=\Lambda(t,s)^{-1} \de \Lambda(t,s)\in \mathbb{R}  ^3 $ and $\bpsi (t,s)= \Lambda(t,s)^{-1} \de \mathbf{r} (t,s) \in \mathbb{R}  ^3 $ are arbitrary functions vanishing at $t=0,T$.
Hamilton's principle \eqref{HP_rod} induces the variational principle
\begin{equation}\label{Reduced_HP_rod} 
\delta \int_0^T\ell( \boldsymbol{\omega} , \boldsymbol{\gamma} , \boldsymbol{\Omega} , \boldsymbol{\Gamma} )\mbox{d}t=0,
\end{equation} with respect to the constrained variations \eqref{delta1}, \eqref{delta2}, which in turn yields the reduced Euler-Lagrange equations
\begin{equation}
\left\lbrace\begin{array}{l}
\displaystyle\lp \prt_t + \bom\times\rp\dede{\ell}{\bom}+\bgam\times\dede{\ell}{\bgam} +\lp\prt_s + \bOm\times\rp\dede{\ell}{\bOm} +\bGam\times \dede{\ell}{\bGam} =0\\
\displaystyle\lp \prt_t + \bom\times\rp\dede{\ell}{\bgam} + \lp\prt_s + \bOm\times\rp  \dede{\ell}{\bGam}=0,
\end{array}\right.
\end{equation} 
together with the boundary conditions
\begin{equation}\label{BC_rod} 
\left.\frac{\delta \ell}{\delta \boldsymbol{\Omega} }\right |_{s=0,L}=0 \, , \quad \left.\frac{\delta \ell}{\delta \boldsymbol{\Gamma } }\right |_{s=0,L}=0.
\end{equation} 
If one of the extremity (say $s=0$) of the rod is kept fixed, \emph{i.e.},  $ \mathbf{r} (t,0)=\mathbf{r} _0 $, $ \Lambda (t,0)= \Lambda _0 $ for all $t$, then only the boundary condition at $s=L$ arises above.

From their definition \eqref{def_conv_var}, the convective variables verify the compatibility conditions
\begin{equation}\label{compat_rod} 
 \partial _t \boldsymbol{\Omega} = \boldsymbol{\omega} \times \boldsymbol{\Omega} +\partial _s  \boldsymbol{\omega}\quad\text{and}\quad  \partial _t \boldsymbol{\Gamma} + \boldsymbol{\omega} \times \boldsymbol{\Gamma} = \partial _s \boldsymbol{\gamma} + \boldsymbol{\Omega} \times \boldsymbol{\gamma}.
\end{equation}

\begin{remark}[Lagrangian reduction by symmetry]\rm  The process of passing from the Lagrangian (or material) representation in terms of $ (\Lambda , \dot{ \Lambda }, \mathbf{r} , \dot{ \mathbf{r} })$ with variational principle \eqref{HP_rod} to the convective representation in terms of $(\boldsymbol{\omega} , \boldsymbol{\gamma} , \boldsymbol{\Omega} , \boldsymbol{\Gamma})$ with constrained variational principle \eqref{Reduced_HP_rod} can be understood via a Lagrangian reduction process by symmetries. It has been carried out in \cite{ElGBHoPuRa2010} and is based on the affine Euler-Poincar\'e reduction theory of \cite{GBRa2009}.  
\end{remark}

\begin{remark}[On the functional form of the Lagrangian] \rm Note that in all our theoretical considerations we will keep the Lagrangian in the general form, as we are interested in the symmetry-reduction approach to the fully three dimensional problem, rather than in the derivation of the equations of motion in a particular reduced setting,  \emph{e.g.}, restricted to two dimensions, straight line, \emph{etc}. We believe that such an approach based on Lagrangian mechanics yields the simplest possible treatment of the elastic, three-dimensional deformation of the tube. 
\end{remark}

\subsection{Definition of the configuration space for the tube with expandable walls conveying fluid}

We now incorporate the motion of the fluid inside the tube, the vorticity and the  motion of the wall of the tube  by  extending  the geometrically exact framework.
Recall that the geometrically exact rod (without fluid) consists of a left invariant system and, therefore, can be written in terms of convective variables. On the other hand, the fluid is a right invariant system, naturally written in terms of spatial variables.
The coupling of these two systems therefore yields the interesting combined fluid-structure interaction involving both convective and spatial variables but whose left and right invariances are broken by the coupling constraint. Without the vorticity, our derivation in this chapter is equivalent to \cite{FGBPu2019}, see also \cite{FGBPu2019a} for a more detailed and pedagogical exposition. 

In addition to the rod variables $(\Lambda , \mathbf{r}) \in \mathcal{F} ([0,L], SO(3) \times \mathbb{R}  ^3 )$ considered above, the configuration manifold for the fluid-conveying tube also contains the Lagrangian description of the fluid. It is easier to start by defining the back-to-label map, which is an embedding $ \psi :[ 0,L] \rightarrow \mathbb{R}  $, assigning to a current fluid label particle $s \in [0,L]$ located at $ \mathbf{r} (s) $ in the tube, its Lagrangian label $s_0 \in \mathbb{R}  $. Its inverse $ \varphi := \psi ^{-1}$ 
gives the current configuration of the fluid in the tube. A time dependent curve of such maps thus describes the fluid motion in the tube, \emph{i.e.}, 
\[
s= \varphi (t,s_0), \quad s \in [0,L].
\]
We now include the motion of the wall of the tube, as a reaction to the fluid motion and pressure. In order to incorporate this effect in the simplest possible case of a tube with a circular cross-section, let us consider the tube radius $R(t,s)$ to be a free variable. In this case, the Lagrangian depends on  $R$, as well as on its time and space derivatives $\dot R$ and $R'$, respectively. Let us assume that $R$ can lie on an interval $I_R$, for example, $I_R=\mathbb{R}_+$ (the set of positive numbers). As we will see below, the presence of the swirl will add an additional scalar function to the variables of the Lagrangian, namely, the  circulation of fluid velocity measured at a given cross-section. 
Then, 
the configuration manifold for the fluid-conveying tube is given by the infinite dimensional manifold
\begin{equation}\label{config_garden_hose} 
\begin{aligned} 
\mathcal{Q} :=\mathcal{F}  &\left( [0,L], SO(3) \times \mathbb{R} ^3  \times I_R   \times \mathbb{R} \right) 
\\ 
& \times 
\left\{ \varphi : \varphi ^{-1} [0,L]  \rightarrow [0,L]\mid \text{$\varphi $ diffeomorphism} \right\}.
\end{aligned} 
\end{equation}
Note that the domain of definition of the fluid motion $s= \varphi (t,s_0)$ is time dependent, i.e., we have $ \varphi (t):[ a(t), b(t)] \rightarrow [0,L]$, for $\varphi (t, a(t))=0$ and $\varphi (t,b(t))=L$.
The time dependent interval $[a(t),b(t)]$ contains the labels of all the fluid particles that are present in the tube at time $t$.

\subsection{Definition of the Lagrangian}\label{Def_Lagrangian}

Let us now turn our attention to the derivation of the Lagrangian of the fluid-conveying geometrically exact tube with expandable walls. The Lagrangian will be computed as the sum of kinetic energy of rod and fluid and negative of the potential energy of the rod's deformation. 
\paragraph{Kinetic energy of the rod.} The kinetic energy of the elastic rod is the function $K_{\rm rod}$ given by
\[
K_{\rm rod}= \frac{1}{2} \int_0^L\left( \alpha | \bgam|^2 + a \dot{R}^2 + \mathbb{I}(R) \bom \cdot \bom \right)| \boldsymbol{\Gamma} | \mbox{d}s,
\] 
where $\alpha$ is the linear density of the tube and $\mathbb{I}(R)$ is the  local moment of inertia of the tube.  The term $ \frac{1}{2}a  \dot{R}^2$ describes the kinetic energy of the radial motion of the tube.

We now derive the total kinetic energy of the fluid. In material representation, the total velocity of the fluid particle with label $s_0$ is given by
\begin{equation}\label{velocity_equalities}
\begin{aligned} 
\frac{d}{dt} \mathbf{r}(t ,\varphi (t,s_0))&= \partial _t \mathbf{r} (t, \varphi (t,s_0))+ \partial _s \mathbf{r} (t, \varphi (t,s_0)) \partial _t \varphi (t,s_0)\\
&=\partial _t \mathbf{r} (t, \varphi (t,s_0))+ \partial _s \mathbf{r} (t, \varphi (t,s_0)) u(t,\varphi (t,s_0)),
\end{aligned}
\end{equation}  
where the Eulerian velocity is defined by
\begin{equation}\label{Eulerian_velocity_u}
u(t,s)=\left(  \partial _t \varphi \circ \varphi ^{-1} \right) (t,s), \quad s \in [0,L].
\end{equation}
Therefore, the  kinetic energy of the translational motion of the fluid reads
\[
K_{\rm fluid, transl}= \frac{1}{2} \int_{ \varphi ^{-1} (0,t)} ^{ \varphi ^{-1} (L,t)} \rho Q_0 (s_0) \left| \frac{d}{dt} \mathbf{r}(t ,\varphi (t,s_0))\right | ^2 \mbox{d}s_0,
\]
where  $Q_0(s_0) \mbox{d} s_0 $ is the initial infinitesimal volume of fluid.

\paragraph{Circulation of velocity and Kelvin's theorem} 
In order to go beyond the theory derived in \cite{FGBPu2019}, we introduce the swirling motion in the Lagrangian. We will introduce the simplest model of swirling. We assume that the shape of the tube at cross-section is approximately circular, and there is an axisymmetric swirling motion in addition to the streamline velocity $u(s,t)$. The trajectories of the fluid particles are locally helical, although they could be quite complex globally, \emph{i.e.} extended for long $t$.   

Take $C_\alpha$ to be a contour that circles a given cross-section close to the boundary at a Lagrangian point $\alpha$. Assuming that the deformations of the centerline remain small,  the cross-section roughly circular, and the viscosity can be neglected. the cross-section remains normal to the centerline. Under this assumption, $C(t)$ will always be close to the boundary of a cross-section at all other times $t$ at the point $s=\varphi(\alpha,t)$. 
For an inviscid incompressible fluid considered here, Kelvin's theorem states that the circulation $\lambda(s,t)$ is conserved along the flow, in other words, for a three-dimensional fluid velocity $\bv$, 
\begin{equation} 
\frac{D \lambda}{Dt} = (\partial_t + u \partial_s) \lambda =0 \, , \quad 
\lambda:= \int_{C(t)}  \bv \cdot \mbox{d} \mathbf{l} =0 \, , 
\label{Kelvin_thm_0} 
\end{equation} 
where $D/D t$ is the full derivative of the integral quantity. A more general version of Kelvin's theorem uses the fluid momentum instead of $ \bf$ \cite{HoSchSt2009}. 
In what follows we shall use $\lambda$ as an additional dynamic quantity in the Lagrangian to model the swirling motion of the flow. 
We thus obtain from \eqref{Kelvin_thm_0}  the following approximate conservation law for the swirling motion. 
\begin{equation} 
(\partial_t + u \partial_s) \lambda =0 
\quad \Leftrightarrow \quad 
\lambda= \lambda_0 \circ \varphi^{-1}(s,t) \, . 
\label{cons_law_lambda} 
\end{equation} 

\begin{remark}[On the description of swirling motion  through vorticity] 
\label{vorticity_remark} 
{\rm 
One could be tempted to introduce a 'typical' vorticity in the streamwise direction $\xi(s,t)$ instead of circulation $\lambda(s,t)$ introduced above. Then, the vorticity advection written in the local coordinates will state 
\begin{equation} 
\partial_t \xi + u \partial_s \xi - \xi \partial_s u=0 \, . 
\label{vorticity_ev} 
\end{equation} 
If the velocity field is assumed to be smooth in every cross-section, then we can write $\lambda \simeq \xi Q$. Differentiating $\lambda$ defined that way, we obtain exactly \eqref{cons_law_lambda}. 
\\ 
While this calculation is formally correct, it is difficult to justify it physically, since the swirling motion in inviscid fluids may have a singularity in the velocities at the centerline of the tube. Thus, the concept of a 'typical' vorticity for a cross-section is ill-defined even for the simplest problem of a perfectly cylindrical straight tube with the flow independent of $s$. In that case, for every cross-section, the problem reduces to the problem of the motion of a point vortex in a circle inviscid fluid, and a point vortex at the center with a given circulation value is a solution \cite{Sa1992,Ne2013}.  However, Kelvin's theorem is still valid,  even when the vorticity is infinite at a set of points, as long as the circulation of the vorticity is finite. Thus, we believe that the circulation $\lambda(s,t)$ defined by \eqref{Kelvin_thm_0} with the evolution equation given by \eqref{cons_law_lambda} is superior characterization of swirl as compared to the notion of a 'typical' vorticity. Note that physical cut-offs must be introduced to utilize the notion of the kinetic energy, as, technically speaking, the kinetic energy of the fluid driven by a point vortex is infinite. 
} 
\end{remark} 
The rotational component of kinetic energy can be then approximated by integrating the rotational component of velocity over a cross-sectional area, and then over the whole tube. Note that dimensionally, $\lambda^2 \simeq R^2 v_{\rm rot}^2$, so $\rho \lambda^2$ already has the \emph{dimensions } of the kinetic energy of rotation per unit length of the tube. The exact value of the kinetic energy depends on the velocity profile of the radial velocity as a function of $r$, namely, $v_r (r,s,t)$. In general, this velocity profile is a complex function of the tube's shape and the local vorticity profile of the fluid. We thus define the profile shape function for every cross-section $C_s$
\begin{equation} 
\Phi=\frac{2 \pi \rho \int_0^R r v_{\rm rot}^2 \mbox{d} r }{ \lambda(s,t)^2  } |\bGam| \,,  \quad 
\Phi=\Phi(\bOm,\bGam, R, R', \ldots). 
\end{equation} 

In general, the shape function $\Phi$ is dimensionless with values of order 1. It describes the deviation of the velocity profile from its given equilibrium value during the dynamics. The profile shape function $\Phi$ can only be obtained through either numerical simulations or experiments, and is thus a modeling component of the theory. Our equations of motion will be derived by the variational principle and will be valid for any $\Phi$. While the exact nature of $\Phi$ and its dependence on the flow is to be investigated in future work, although some conclusions can be drawn from simple considerations. 
For example, for a point vortex on the plane with circulation $\lambda$, the rotational velocity is $v_{\rm rot}= \frac{\lambda}{2 \pi r}$, so the kinetic energy of the fluid's rotation in a circle of radius $R$ is infinite. However, if we introduce a smoothing cut-off of radius $\epsilon$ around the vortex attributed to viscosity, we arrive to 
\begin{equation} 
\Phi= \frac{1}{\lambda^2} \int 2 \pi \rho \int_{\epsilon}^R r v_{\rm rot}^2 \mbox{d} r = \frac{1}{2 \pi} \log \frac{R}{\epsilon} 
\label{Phi_Vortex}
\end{equation}  
In the numerical simulations in this paper we will use \eqref{Phi_Vortex} for a fixed value of $\epsilon$. Then, 
\begin{equation} 
K_{\rm fluid, rot}=\frac{\pi}{2} \int \rho \Phi \lambda^2 \mbox{d}s \, . 
\label{KE_rot_fluid} 
\end{equation} 
Using \eqref{velocity_equalities} together with the change of variables $s= \varphi (t,s_0)$, we can rewrite $K_{\rm fluid}$ as
\[
K_{\rm fluid}= \frac{1}{2} \int_{0} ^{ L} \rho Q \left| \boldsymbol{\gamma} + \boldsymbol{\Gamma} u\right | ^2 + 
  \rho \Phi \lambda^2  \mbox{d}s,
\]
where $\rho$ is the mass density of the fluid per unit volume, in units Mass/Length$^3$, and $Q(t,s)$ is the area of the tube's cross section, in units Length$^2$.

\paragraph{Elastic energy.} The potential energy due to elastic deformation is a function of $\bOm$, $\bGam$ and $R$. While the equations will be derived for an arbitrary potential energy,  we shall assume the simplest possible quadratic expression for the calculations, namely,
\begin{equation}
E_{\rm rod}=\frac{1}{2}\int_0^L\Big( \mathbb{J}\bOm\! \cdot  \!\bOm
+ \lambda(R) |\bGam- \boldsymbol{\chi} |^2 +2F(R,R',R'')\Big)  |\bGam| \mbox{d} s\,,
\label{e_rod} 
\end{equation} 
where $ \boldsymbol{\chi} \in \mathbb{R} ^3 $ is a fixed vector denoting the axis of the tube in the reference configuration, $ \mathbb{J}$ is a symmetric positive definite $3 \times 3$ matrix, which may depend on $R$, $R'$ and $R''$, and $\lambda (R)$ is the stretching rigidity of the tube. 
The stretching term proportional to $\lambda(R)$ can take the more general form $\mathbb{K} (\bGam-  \boldsymbol{\chi} ) \cdot (\bGam-  \boldsymbol{\chi} )$, where $\mathbb{K}$ is a $3 \times 3$ tensor.
The part of this expression for the elastic energy containing the first two terms in \eqref{e_rod} is commonly used for a Cosserat elastic rod, but more general functions of deformations $\bGam $ are possible, see \cite{FGBGePu2018} for a more detailed discussion of possible forms of the potential energy. A particular case is a quadratic  function  of $\bGam$ leading to a linear dependence between stresses and strains. We have also introduced the elastic energy of wall $F(R,R',R'')$ which can be explicitly computed for simple elastic tubes. In general $F$ depends on higher derivatives, such as $R''$. We shall derive the equations of motion in their general form for arbitrary $F$, but for particular examples used in simulations we will use $F(R,R')$ for simplicity.

\paragraph{Mass conservation.} Before we write the final expression for the Lagrangian, let us discuss the question of the mass conservation since it will be used as a constraint to derive the equations of motion, with the appropriate Lagrange multiplier playing the role of fluid's pressure. We shall assume that the fluid fills the tube completely, and the fluid velocity at each given cross-section is aligned with the axis of the tube. 
Since we are assuming a one-dimensional approximation for the fluid motion inside the tube, the mass density per unit length $\rho Q=\rho A |\bGam|$, has to verify
\begin{equation}
\label{eq_xi} 
\rho Q= ( \rho Q _0 \circ \varphi ^{-1} )\partial _s \varphi ^{-1}.
\end{equation} 
Since $\rho=$const, and the fluid is incompressible, we deduce the conservation law
\begin{equation} 
 Q=  Q _0 \circ \varphi ^{-1} \partial _s \varphi ^{-1} \quad 
\Rightarrow \quad \partial _t Q +\partial _s( Q u)=0.
\label{Q_cons}
\end{equation} 
The origin of this conservation law is illustrated on the left panel of Figure~\ref{fig:cons_law_trajectories}. A Lagrangian point $a$ is mapped to $s=\varphi(a,t)$, and the Lagrangian point $a'$ is mapped to $s'=\varphi(a',t)$. Inversely, a  point 
on a centerline $s$ will be mapped to the initial Lagrangian point $a=\varphi^{-1}(s,t)$. Taking $a'=a+\mbox{d}a$, we see that the infinitesimal physical volume 
is described at the physical space by $Q \mbox{d} s$, and the corresponding volume at $t=0$ by $Q_0 (a) \partial_s \varphi^{-1} (s,t)$, with $a=\varphi^{-1}(s,t)$, giving exactly \eqref{Q_cons}. 

The key new feature of this Section is in the introduction of the swirling motion in the kinetic energy part of the Lagrangian. 
The particle trajectories for such a swirling motion are approximately helical, as illustrated by a sketch on the right panel of Figure~\ref{fig:cons_law_trajectories}. 
\begin{figure}[h]
\centering
\includegraphics[width=0.48\textwidth]{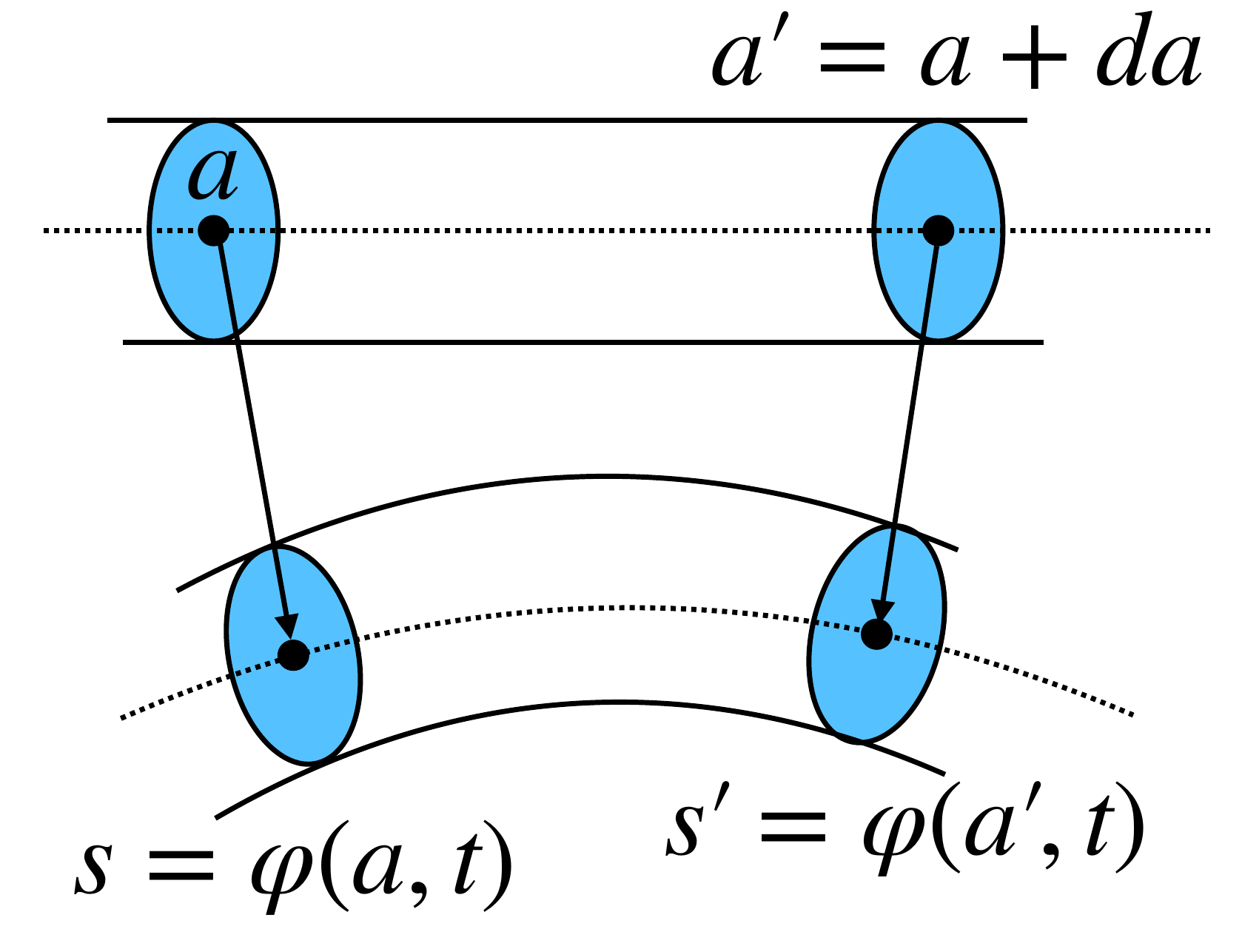}
\includegraphics[width=0.48\textwidth]{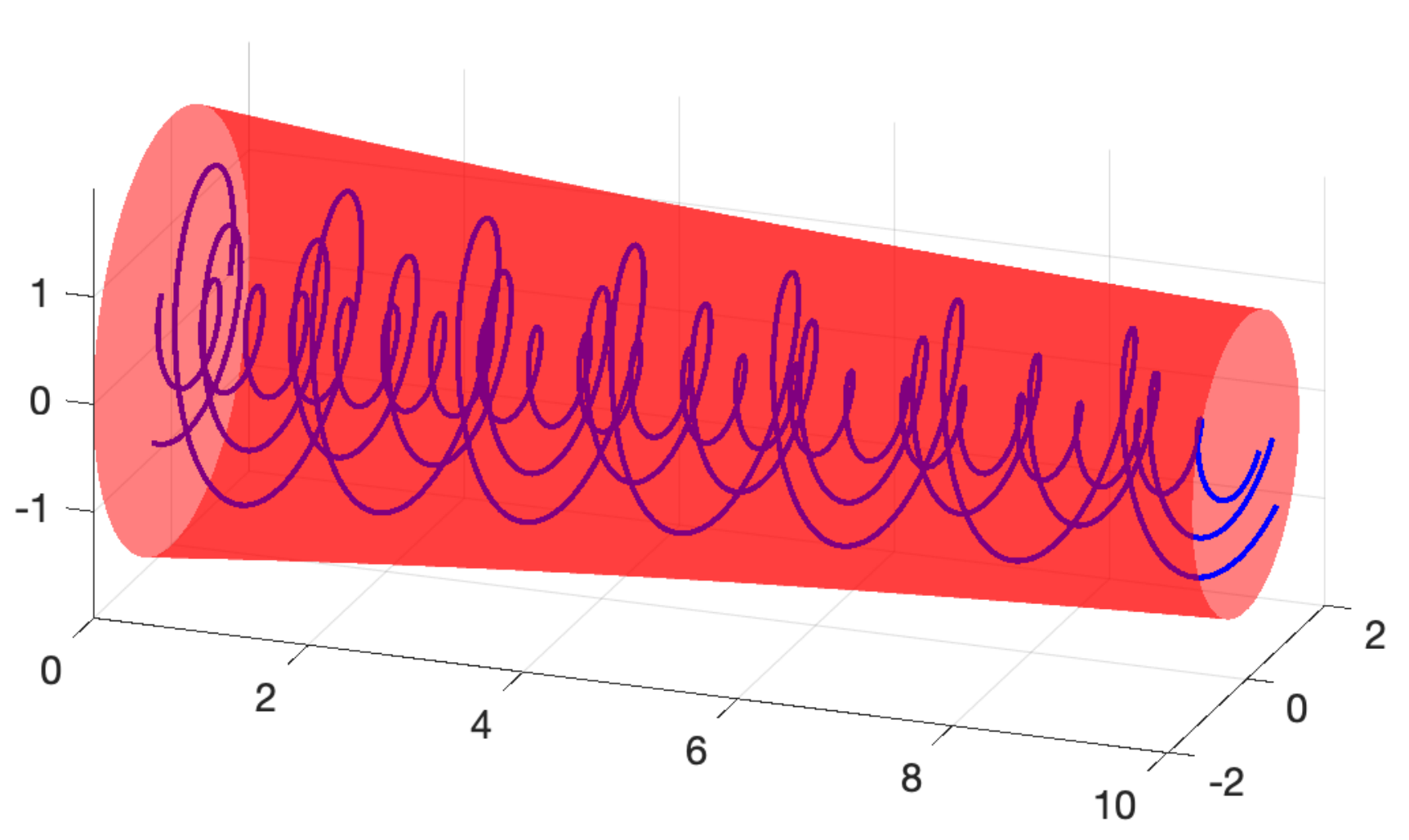}
\caption{Left panel: a sketch of conservation law \eqref{eq_xi}. Right panel: a tube (shown in red) and particle trajectories (shown in blue).  }
\label{fig:cons_law_trajectories}
\end{figure}

\paragraph{Lagrangian.} From all the expressions given above, we obtain the Lagrangian of the fluid-conveying given by
\begin{equation}\label{total_Lagrangian}
\mathsf{L}=\mathsf{L} \big(  \Lambda , \dot {\Lambda }, \mathbf{r} , \dot{\mathbf{r}}, \varphi , \dot{ \varphi },R,\dot{R}\big) :T\mathcal{Q}  \rightarrow \mathbb{R},\quad \mathsf{L}=K_{\rm rod}+K_{\rm fluid}-E_{\rm rod} - E_{\rm int},
\end{equation}
and defined on the tangent bundle $T \mathcal{Q} $ of the configuration space $ \mathcal{Q} $, see \eqref{config_garden_hose}. 
Note that all the arguments of $\mathsf{L}$ are functions of $s$, so we don't need to include the spatial derivatives of $(\Lambda, \mathbf{r}, R)$ explicitly as variables in $\mathsf{L}$. These spatial derivatives appear explicitly in the expression of the integrand of the reduced Lagrangian \eqref{Lagrangian_fluid_tube} below. Assuming there is a uniform external pressure $p_{\rm ext}$ acting on the tube, the Lagrangian expressed 
in terms of the variables $\boldsymbol{\omega} ,\boldsymbol{\gamma} , \boldsymbol{\Omega}, \boldsymbol{\Gamma} ,u,  Q, \lambda $ reads
\begin{equation}\label{Lagrangian_fluid_tube}
\begin{aligned}
&\ell  ( \boldsymbol{\omega} ,\boldsymbol{\gamma} , \boldsymbol{\Omega} ,  \boldsymbol{\Gamma} ,u, Q , \lambda, R,\dot{R}) 
\\
&=  \int_0^L \Big[\Big(\frac{1}{2} \alpha | \bgam|^2 +  \frac{1}{2}\mathbb{I}(R) \bom\! \cdot\! \bom + \frac{1}{2} a  \dot R^2 - F(R,R') - \frac{1}{2} \mathbb{J}  \bOm \!\cdot \! \bOm\\
& \qquad \qquad  - \frac{1}{2} S(R) |\bGam- \boldsymbol{\chi} |^2 \Big) | \bGam|
+ \frac{1}{2}\rho Q  \left| \boldsymbol{\gamma} + \boldsymbol{\Gamma} u\right | ^2+  \frac{1}{2} \rho \Phi \lambda^2 s   -p_{\rm ext} Q\Big] \mbox{d} s 
\\ 
&=: 
\int_0^L \Big[\ell_0( \boldsymbol{\omega} ,\boldsymbol{\gamma} , \boldsymbol{\Omega} ,  \boldsymbol{\Gamma} ,u,Q,R,\dot{R}, R',\lambda)  -p_{\rm ext} Q \Big]  \mbox{d} s \, .
\end{aligned}
\end{equation}
where $S(R)$ is the function describing the potential energy of tube's stretching along its centerline as a function of $R$. Our derivation of the equations of motion \eqref{system_ell} below will be valid for the Lagrangian $\ell$ being an arbitrary function of its variables. Of course, the numerical simulations will require specifying the exact functional form of the Lagrangian $\ell$ for every particular case. 

\subsection{Variational principle and equations of motion}

The equations of motion are obtained from the Hamilton principle applied to the Lagrangian \eqref{total_Lagrangian}, namely
\begin{equation}\label{HP_total} 
\delta \int_0^T\mathsf{L}( \Lambda , \dot {\Lambda }, \mathbf{r} , \dot{\mathbf{r}}, \varphi, \dot\varphi, \lambda, R, \dot R )\mbox{d}t=0, 
\end{equation} 
for arbitrary variations $\delta\Lambda, \delta\mathbf{r}, \delta\varphi, \delta R$ vanishing at $t=0,T$. In terms of the symmetry reduced Lagrangian $\ell$, this variational principle becomes
\begin{equation} 
\de  \int_0^T\ell ( \boldsymbol{\omega} ,\boldsymbol{\gamma} , \boldsymbol{\Omega} ,  \boldsymbol{\Gamma} ,u, \xi , \lambda, R,\dot{R}) \mbox{d} t =0 
\, , 
\label{min_action_gas} 
\end{equation} 
for variations \eqref{delta1}, \eqref{delta2}, and $\de u$, $\de \xi$ and $\de \lambda$ computed as
\begin{align} 
\de u & =  \partial_t\eta + u \partial_s\eta - \eta \partial _su  
\, \quad 
\delta \lambda= -\eta \partial _s \lambda,\label{delta_u_lam} \, , 
\quad 
\de (Q_0 \circ \varphi^{-1} \partial_s \varphi^{-1})  = - \partial_s (\eta Q ) 
\end{align} 
where $\eta = \de \varphi \circ \varphi^{-1}$. Note that $\eta(t,s)$ is an arbitrary function vanishing at $t=0,T$. Note that the Lagrangian function in \eqref{min_action_gas} will contain the incompressibility constraint \eqref{Q_cons} with the Lagrange multiplier $p$, which we drop from the notation in \eqref{min_action_gas} for simplicity. A lengthy computation (see \cite{FGBPu2019} for details) yields the system 
\begin{equation}\label{system_ell} 
\left\lbrace\begin{array}{l}
\displaystyle\vspace{0.2cm}\frac{D}{Dt} \dede{\ell}{\bom}+\bgam\times\dede{\ell}{\bgam} +\frac{D}{Ds}  \dede{\ell}{\bOm}   +\bGam\times \dede{\ell}{\bGam}=0\\
\displaystyle\vspace{0.2cm}\frac{D}{Dt} \dede{\ell}{\bgam} +\frac{D}{Ds}  \dede{\ell}{\bGam} =0\\
\displaystyle\vspace{0.2cm}\prt_t\frac{\delta \ell}{\delta u}  + u \partial _s\frac{\delta \ell}{\delta u}+ 2  \frac{\delta \ell}{\delta u} \partial _s u  + \frac{\delta \ell}{\delta \lambda} \partial _s \lambda = Q \partial _s \frac{\delta \ell}{\delta Q } \\
\displaystyle\vspace{0.2cm} 
\partial_t \frac{\delta \ell }{\delta \dot R} - \frac{\delta \ell}{\delta R} =0 
\\
\displaystyle\vspace{0.2cm} \partial _t \boldsymbol{\Omega} = \boldsymbol{\Omega} \times \boldsymbol{\omega} +\partial _s  \boldsymbol{\omega}, \qquad  \partial _t \boldsymbol{\Gamma} + \boldsymbol{\omega} \times \boldsymbol{\Gamma} = \partial _s \boldsymbol{\gamma} + \boldsymbol{\Omega} \times \boldsymbol{\gamma}\\
\displaystyle \partial _t  Q + \partial _s (Q u)=0, \qquad\partial _t \lambda+ u \partial _s \lambda=0,
\end{array}\right.
\end{equation}
where the symbols $  {\delta \ell}/{\delta \boldsymbol{\omega} },  {\delta \ell}/{\delta \boldsymbol{\Gamma} }, ... $ denote the functional derivatives of $\ell$ relative to the $L ^2 $ pairing, and we introduced the notations
\[
\frac{D}{Dt} =\partial_t + \boldsymbol{\omega}\times\quad\text{and}\quad\frac{D}{Ds} =\partial_s + \boldsymbol{\Omega}\times.
\]
Note that the first equation arises from the terms proportional to $\bsigma$ in the variation of the action functional and thus describes the conservation of angular momentum.  The second equation arises from the terms proportional to $\boldsymbol{\psi}$ and describes the conservation of linear momentum. We remind the reader that the quantities $\bsigma$ and $\boldsymbol{\psi}$ are defined by $ \bsigma=(\Lambda ^T \de \Lambda)^\vee$ and 
$\boldsymbol{\psi}= \Lambda^T \de \mathbf{r}$, with $(\bsigma(s,t), \boldsymbol{\psi}(s,t))$ for fixed $s$ and $t$ taking   values in the Lie algebra $\mse(3)$. 

The third equation is obtained by collecting the terms proportional to $\eta$ and describes the conservation of fluid momentum.  The fourth equation comes from collecting the terms proportional to $\de R$ and describes the elastic deformation of the walls due to the pressure. Finally, the last four equations arise from the four definitions $\bOm=(\Lambda^{-1} \Lambda_s)^\vee$, $\bGam=\Lambda^{-1} \mathbf{r} '$, $Q =  ( Q _0 \circ \varphi ^{-1} )\partial _s \varphi ^{-1}$, and $\lambda= \lambda_0 \circ \varphi ^{-1}$.

We will derive the equations for general $Q=Q(\bOm,\bGam,R)$, keeping in mind the fully three-dimensional dynamics, but will only use it for the tubes satisfying $Q=\pi R^2 |\bGam|$ which is a good approximation for mostly straight tubes.  For the choice
\[
\ell
=
\int_0^L \Big[\ell_0( \boldsymbol{\omega} ,\boldsymbol{\gamma} , \boldsymbol{\Omega} ,  \boldsymbol{\Gamma} ,u, \lambda, R,\dot{R}, R') - p( Q-Q_0 \circ \varphi^{-1} \partial_s \varphi^{-1}) 
 - p_{\rm ext}Q \Big] \mbox{d} s \, ,
\]
the functional derivatives are computed as
\begin{equation}\label{functional_derivatives}
\begin{aligned}
\frac{\delta \ell}{\delta \boldsymbol{\Omega} }&= \frac{\partial  \ell_0}{\partial \boldsymbol{\Omega} } +  (p- p_{\rm ext}) \frac{\partial Q}{\partial \boldsymbol{\Omega} }\\
\frac{\delta \ell}{\delta \boldsymbol{\Gamma } }&= \frac{\partial \ell_0}{\partial \boldsymbol{\Gamma } } +  (p- p_{\rm ext}) \frac{\partial Q}{\partial \boldsymbol{\Gamma } }\\
\frac{\delta \ell}{\delta R }&= \frac{\partial \ell_0}{\partial R } - \partial _s\frac{\partial \ell_0}{\partial R' }+ \partial _ s^2 \frac{\partial \ell_0}{\partial R'' } +  (p- p_{\rm ext}) \frac{\partial Q}{\partial R }\\
\frac{\delta \ell}{\delta Q }&= \frac{\partial \ell_0}{\partial Q } - (p-p_{\rm ext}) ,
\end{aligned}
\end{equation}
where $\partial \ell_0/\partial \boldsymbol{\Omega}, \partial \ell_0/\partial \boldsymbol{\Gamma} $, ... denote the ordinary partial derivatives of $\ell_0$,  whose explicit form can be directly computed from the expression of $\ell_0$ in \eqref{Lagrangian_fluid_tube}.

\begin{theorem}\label{vp_tube_fluid}  For the Lagrangian $\ell$ in \eqref{Lagrangian_fluid_tube}, the variational principle \eqref{min_action_gas} with constrained variations \eqref{delta1}, \eqref{delta2}, \eqref{delta_u_lam} yields the equations of motion
\begin{equation}\label{full_3D}{\fontsize{11pt}{9pt}\selectfont
\!\!\!\!\left\lbrace\!\begin{array}{l}
\displaystyle\vspace{0.2cm}\frac{D}{Dt} \pp{\ell_0}{\bom}+\bgam\times\pp{\ell_0}{\bgam} +\frac{D}{Ds}  \left( \pp{\ell_0}{\bOm} + (p-p_{\rm ext})\pp{Q}{\bOm} \right)   +\bGam\times 
 \left( 
\pp{\ell_0}{\bGam} + (p-p_{\rm ext})\pp{Q}{\bGam} 
\right) =0
\\
\displaystyle\vspace{0.2cm}\frac{D}{Dt} \pp{\ell_0}{\bgam} +\frac{D}{Ds} \left( \pp{\ell_0}{\bGam}+(p-p_{\rm ext}) \pp{Q}{\bGam}   \right)=0\\
\displaystyle\vspace{0.2cm}\prt_t\frac{ \partial  \ell_0}{ \partial  u}  + u \partial _s\frac{ \partial  \ell_0}{ \partial  u}+ 2  \frac{ \partial  \ell_0}{ \partial  u} \partial _s u +\pp{\ell_0}{\lambda} \partial_s \lambda  = - Q \partial_s p \\
\displaystyle\vspace{0.2cm} 
\partial_t \pp{\ell_0}{\dot R} - \partial^2_s \pp{\ell_0}{ R''} +  \partial_s \pp{\ell_0}{ R'} -  \pp{\ell_0}{ R}- (p-p_{\rm ext})\frac{\partial Q}{\partial R}=0 
\\
\displaystyle\vspace{0.2cm} \partial _t \boldsymbol{\Omega} = \boldsymbol{\Omega} \times \boldsymbol{\omega} +\partial _s  \boldsymbol{\omega}, \qquad  \partial _t \boldsymbol{\Gamma} + \boldsymbol{\omega} \times \boldsymbol{\Gamma} = \partial _s \boldsymbol{\gamma} + \boldsymbol{\Omega} \times \boldsymbol{\gamma}\\
\displaystyle \partial _t Q+ \partial _s (Qu)=0, \qquad\partial _t \lambda + u \partial _s \lambda =0
\end{array}\right.}
\end{equation}
together with appropriate boundary conditions enforcing vanishing of the variations of the boundary terms. 
\rem{ 
\begin{equation}\label{BC} 
\left. \dede{\ell}{\bOm} - \mu \pp{Q}{\bOm} \right|_{s=0,L} =0 \, , 
\quad 
\left. \dede{\ell}{\bGam} - \mu \pp{Q}{\bGam} \right|_{s=0,L}=0 \, , \quad 
\left. \dede{\ell}{u}u - \mu Q \right|_{s=0,L}=0
\end{equation} 
hold. If one of the extremity of the rod is fixed, say $s=0$, then the first two boundary conditions only arise at $s=L$. If the fluid velocity is prescribed at one extremity of the rod, say $s=0$, then the last boundary condition only arise at $s=L$.
} 
\end{theorem} 
\begin{proof}  The system is obtained by replacing the expression of the functional derivatives \eqref{functional_derivatives} in the system \eqref{system_ell}. \end{proof} 

Note that the term $\dede{\ell_0}{R}$ contains both the terms coming from the elastic energy  and the derivatives of the shape function $\Phi(R)$. Thus, the vorticity contributes to the pressure term and momentum balance of the wall, as expected. 

\begin{remark}[On the formal correspondence between the entropy and circulation] 
{\rm One can notice the formal correspondence between \eqref{system_ell} and the corresponding equations for the incompressible fluid with thermal energy derived in \cite{FGBPu2019}. Mathematically, the \emph{appearance} of equations is equivalent when the circulation $\lambda$ is associated with entropy $S$ in \cite{FGBPu2019}, which is due to the same way the circulation and entropy are advected. That association is purely mathematical: there is no physical correspondence between these quantities, and they enter the Lagrangian in a very different way. The circulation enters the kinetic energy part of the Lagrangian, and entropy enters the internal energy part. Thus, in spite of the apparent formal similarity between these equations, the use of any particular physically relevant Lagrangian leads to very different equations for the fluids containing  swirling vs thermal energy.  }
\end{remark} 
\medskip

Equations \eqref{full_3D}  have to be solved as a system of nonlinear partial differential equations, since all equations are coupled through the functional form of the Lagrangian \eqref{Lagrangian_fluid_tube}. This can be seen, for example, from computing the derivative 
\begin{equation} 
\frac{ \partial \ell _0 }{ \partial  u} = \rho A(R) | \boldsymbol{\Gamma} |  \big( \bgam + \bGam u \big)\cdot \bGam\, ,
\label{dldu}
\end{equation}
which appears in the the balance of fluid momentum, \emph{i.e.}, the third equation in \eqref{full_3D}. Thus, the fluid momentum involves the radius, stretch $\bGam$, and linear velocity $\bgam$, as well as the fluid's velocity. 

If we define the rescaled fluid momentum as
\[
m:= \frac{1}{\rho Q}\dede{\ell_0}{u} 
\]
the third equation in \eqref{full_3D} can be (deceptively) simply written as
\begin{equation}
\label{meq}
\partial _t m+  \partial _s \left(mu  -  \frac{1}{\rho}\pp{\ell_0}{Q} \right)+\frac{1}{\rho Q}  \pp{\ell_0}{\lambda} \lambda_s=-\frac{1}{\rho} \partial_s p \, , \quad  \, . 
\end{equation}
which is reminiscent of the equations of motion for one-dimensional fluid with the additional circulation terms describing the evolution of vorticity. 
\rem{ 
Also, in \eqref{meq} we have
\begin{equation} 
\label{dell_dxi} 
\frac{1}{\rho}\pp{\ell_0}{Q} = \frac{1}{2}  \left| \boldsymbol{\gamma} + \boldsymbol{\Gamma} u\right | ^2, \quad 
\frac{1}{\rho Q }\pp{\ell_0}{\lambda} =\frac{\Phi \lambda}{Q}  \, . 
\end{equation} 

\medskip
} 

For $Q(R,\bGam)=A(R) |\bGam|$, the equations \eqref{full_3D} can be further simplified since 
\[ 
\pp{Q}{\bGam} = A(R) \frac{\bGam}{|\bGam|} =Q  \frac{\bGam}{|\bGam|^2} \quad\text{and}\quad   \bGam \times \pp{Q}{\bGam}=0  \, . 
\]
Notice also that for $A(R)= \pi R ^2 $, we have $\mbox{d}A/\mbox{d}R= 2 \pi R$, the circumference of a circle.

\section{Tubes with a straight centerline: dynamics and reduction to a single equation} 
\label{sec:1D_tubes}

In this section we shall specify our model to the case of an inextensible and unshearable expandable tube. We then focus on straight expandable tubes with no rotational motion and compare our model with previous works. This case has been studied extensively in the literature, mostly in the context of blood flow involving incompressible fluid. We also show that these simplified models arise from a variational principle using the back-to-labels map. The explicit use of back-to-labels map allows to derive a single equation of motion without the need for the pressure-like Lagrange multiplier term $p$.

\subsection{Equations of motion for inextensible and unshearable tubes}

The dynamics of an inextensible and unshearable, but expandable, tube conveying compressible fluid is obtained by imposing the constraint $\bGam (t,s)= \bchi$, for all $t,s$. If we denote by $\mathbf{z}$ the Lagrange multiplier associated to this constraint and add the term $\int_0^T\int_0^L \mathbf{z} \cdot ( \boldsymbol{\Gamma} - \bchi)\mbox{d}s\,\mbox{d}t$
to the action functional in our variational principle \eqref{min_action_gas}, the first two equations in \eqref{full_3D} will change to
\[
\left\lbrace
\!\!\begin{array}{l}
\displaystyle\vspace{0.2cm}\frac{D}{Dt} \pp{\ell_0}{\bom}+\!\bgam\!\times\pp{\ell_0}{\bgam} +\frac{D}{Ds}  \left( \pp{\ell_0}{\bOm} +  (p-p_{\rm ext})  \pp{Q}{\bOm}  \right)   +\!\bGam\!\times 
 \left( 
\pp{\ell_0}{\bGam} + (p-p_{\rm ext})  \pp{Q}{\bGam} + \mathbf{z} 
\right) =0
\\
\displaystyle\frac{D}{Dt}\pp{\ell_0}{\bgam} + \frac{D}{Ds}\left( \pp{\ell_0}{\bGam}+ (p-p_{\rm ext})  \pp{Q}{\bGam}   + \mathbf{z}\right)=0.
\end{array}\right.
\]
The physical meaning of $\mathbf{z}$ is the reaction force enforcing the inextensibility constraint. 

Particular simple solutions of this system can be obtained by assuming the axis of the tube being straight and no rotational motion. In that case, the inextensibility constraint leads to $\mathbf{z} = z \bchi$, the direction along the axis. 
For such particular solutions, the angular momentum equation is satisfied identically, and the linear momentum equation reduces to a one dimensional equation for the reaction force $z$. Therefore, we only need to compute the fluid momentum equation and the Euler-Lagrange equations for $R$. From the third and fourth equations in \eqref{full_3D}, we get  for $A=A(R)=\pi R^2$: 
\begin{equation}\label{inext_unshear_omega}  
\left\lbrace\begin{array}{l}
\displaystyle\vspace{0.2cm}\partial _t \big(\rho A u \big)+\partial _s \big(\rho A u^2   \big)
+  \rho  \Phi  \lambda \lambda_s   =- A \partial _s p 
\\
\displaystyle a  \ddot R- \partial _s  \pp{F}{R'} + \pp{F}{R} -\frac{1}{2} \pp{\Phi}{R} \lambda^2 = 2 \pi R \left( p - p_{\rm ext} \right),
\end{array} \right.
\end{equation} 
together with the conservation of mass and vorticity 
\begin{equation} 
\partial _t A +\partial _s( A u)=0, \quad \partial _t \lambda +u \partial _s \lambda=0 ,
\label{lambda_A_cons}
\end{equation}
This gives a system of four equations for the four variables $u$, $R$, $\lambda$, and $ p$. 

It is often more convenient to use the cross-sectional area $A$ instead of the radius $R$, and express the Lagrangian
in terms of $A$ and its derivatives, as well as $u$ and $\lambda$. To explicitly denote the change of variables, we shall denote that Lagrangian as $\ell_A$, and the Lagrangian function as $\ell_{0,A}$, so $\ell_A=\int \ell_{0,A} \mbox{d} s$. The equations of motion are then: 
\begin{equation} 
\label{inext_unshear_A}  
\left\lbrace
\begin{array}{l}
\displaystyle\vspace{0.2cm}\partial _t \pp{\ell_{0,A}}{u}+ u \partial _s \pp{\ell_{0,A}}{u} +2 \pp{\ell_{0,A}}{u} u_s = - A p_s 
\\
\displaystyle\vspace{0.2cm} \partial_t \pp{\ell_A}{A_t} + \partial _s  \pp{\ell_A }{A_s}  - \pp{\ell_A }{A} +\pp{\ell}{\lambda} \lambda_s  = \left( p - p_{\rm ext} \right),
\\ 
\displaystyle A_t + \partial_s (A u ) = 0 \, , \quad \lambda_t + u \lambda_s =0 \, . 
\end{array} \right.
\end{equation} 

\rem{ 
For example, the Lagrangian can be taken as the sum of the kinetic energy of the wall, the kinetic energy of the fluid, and minus the potential energy of the wall deformation, written as 
\begin{equation} 
\ell_{0,A}=\int K_A(A, A_t) + \frac{1}{2} \rho A u^2 + \frac{1}{2} \Phi_A(A) \lambda^2 - F_A(A,A_s)  \mbox{d} s
\label{ell_A_explicit} 
\end{equation} 
where $K_A(A, A_t)$ is the kinetic energy of the wall motion expressed in terms of $A$ and its time derivatives, $F_A(A,A_s)$ is the elastic energy function expressed in terms of $A$ and its spatial derivatives, and $\Phi_A$ is the vorticity shape function $\Phi(R)$ expressed in terms of $A$. 
 Equations \eqref{inext_unshear_omega} then become 
\begin{equation}\label{inext_unshear_A}  
\left\lbrace
\begin{array}{l}
\displaystyle\vspace{0.2cm}\partial _t \big(\rho A u \big)+\partial _s \big(\rho A u^2   \big)
+  \rho  \Phi  \lambda \lambda_s   =- A \partial _s p 
\\
\displaystyle \partial_t \pp{K_A}{A_t} + \partial _s  \pp{F_A}{A_s}  - \pp{}{A}(F_A-K_A) -\frac{1}{2} \pp{\Phi}{A} \lambda^2 = \left( p - p_{\rm ext} \right),
\end{array} \right.
\end{equation} 
which have to be solved together with \eqref{lambda_A_cons}. 
} 
Without the vorticity, \eqref{inext_unshear_A} represent the familiar model of tube with elastic walls incorporating the wall inertia \cite{QuTuVe2000,FoLaQu2003}. For example, one can consider the following form of the Lagrangian function $\ell_{0,A}:$ 
\begin{equation} 
\ell_{0,A}=K_A(A, A_t) + \frac{1}{2} \rho A u^2 + \frac{1}{2} \Phi_A(A) \lambda^2 -2  F_A(A,A_s)  \, , 
\label{ell_A_explicit} 
\end{equation} 
where $K_A(A, A_t)$ is the kinetic energy of the wall motion expressed in terms of $A$ and its time derivatives, $F_A(A,A_s)$ is the elastic energy function expressed in terms of $A$ and its spatial derivatives, and $\Phi_A$ is the vorticity shape function $\Phi(R)$ expressed in terms of $A$. 
If the dependence of the Lagrangian function $\ell_{0,A}$ on $A_t$ and $A_s$ is neglected, the second equation of \eqref{inext_unshear_A} yields the so-called \emph{wall pressure law} giving an explicit connection between pressure and area, see \cite{Se2016} and references therein. This wall pressure law $p(A)$ can then be substituted in the first equation of \eqref{inext_unshear_A}, and, in conjunction with the conservation law (the first equation of \eqref{lambda_A_cons}) allows  to obtain a closed set of two coupled equations for $A$ and $u$. We shall not drop the inertia terms in \eqref{inext_unshear_A} in what follows and pursue the more general case of a general Lagrangian dependence on $A$ and its time and $s$-derivatives. 

\subsection{The case of uniform circulation, energy, and the new constant of motion} 
\label{sec:uniform_lam} 
Let us now analyze the equations of motion \eqref{inext_unshear_A} in the case when $\lambda=\lambda_*=$const. Then, $\lambda=\lambda_*$ everywhere in the flow. Physically, such system can be realized by producing a constant vortex circulation at tube's entry using a mechanical device, \emph{e.g.} an idealized propeller. In that case, we can define $\overline{\ell}_{0,A}=\ell_{0,A}$ with $\lambda$ set to be equal to the value of $\lambda_*$. Equations of motion for the general Lagrangian function $\overline{\ell}_{0,A}(A,A_t,A_s,u)$ are then given as 
\begin{equation} 
\left\lbrace
\begin{array}{l}
\displaystyle\vspace{0.2cm}\partial _t \pp{\overline{\ell}_{0,A}}{u}+ u \partial _s \pp{\overline{\ell}_{0,A}}{u} +2 \pp{\overline{\ell}_{0,A}}{u} u_s = - A p_s 
\\
\displaystyle\vspace{0.2cm} \partial_t \pp{\overline{\ell}_{0,A}}{A_t} + \partial _s  \pp{\overline{\ell}_{0,A}}{A_s}  - \pp{\overline{\ell}_{0,A}}{A}   = \left( p - p_{\rm ext} \right),
\\ 
\displaystyle A_t + \partial_s (A u ) = 0 \, . 
\end{array}
\right. 
\label{eq_A} 
\end{equation} 
One would expect that there is an energy-like quantity $E$ which is conserved on solutions of \eqref{eq_A}. Indeed, we can prove the following 
\begin{lemma}[Energy conservation] 
{\rm 
The quantity 
\begin{equation} 
E= \int \pp{\overline{\ell}_{0,A}}{A_t} A_t + \pp{\overline{\ell}_{0,A}}{u} u -  \overline{\ell}_{0,A} \mbox{d} s \, . 
\label{E_quantity} 
\end{equation} 
is conserved on solutions of \eqref{eq_A}. 
}
\end{lemma} 
\begin{proof} 
This result is proven by computing the time derivative of $E$ and integrating by parts, using the equations of motion \eqref{eq_A}, assuming that the boundary terms vanish when integrating by parts, and utilizing the constraint, \emph{i.e.}, the last equation of \eqref{eq_A}. This proof closely follows that of Lemma~\ref{lemma:energy} below, and it is relatively straightforward, so we do not present is here. 
\end{proof} 

\begin{remark}[On the physical meaning of $E$] 
{\rm 
One can notice that the quantity $E$ defined by \eqref{E_def} is, in general, \emph{not} equal to the kinetic plus potential energy (also sometimes called the 'physical' energy).  For example, for the Lagrangian defined by \eqref{ell_A_explicit}, 
\begin{equation} 
E_a= \frac{1}{2} \int K_A(A, A_t) + \frac{1}{2} \rho A u^2 \fbox{-} \frac{1}{2} \Phi_A(A) \lambda^2 + 2 F_A(A,A_s)  \mbox{d} s\, ,  
\label{E_a} 
\end{equation} 
where as the kinetic plus potential energy will be given by the expression 
\begin{equation} 
\hspace{-3mm} 
 \mbox{Kin $+$ Pot Energy}= \frac{1}{2} \int K_A(A, A_t) + \frac{1}{2} \rho A u^2 \fbox{+} \frac{1}{2} \Phi_A(A) \lambda^2 + 2 F_A(A,A_s)  \mbox{d} s \, . 
\label{E_a_wrong} 
\end{equation} 
The energy defined by \eqref{E_a} is conserved, whereas the quantity defined by \eqref{E_a_wrong} is not. Changing time derivatives to the corresponding momenta in \eqref{E_a} will give the Hamiltonian of the system. The non-conservation of 'physical' energy, \emph{i.e.}, the kinetic plus potential energy,  is typical for systems with rotation, see \cite{JoSa2000}.  We believe it is difficult to guess the rather non-intuitive form \eqref{E_a} without appealing to the variational principle used in the derivation of \eqref{single_b_eq_non_dim}. 
} 
\end{remark} 

What is rather surprising, however, is that one can find there is another constant of motion. We prove the following 
\begin{theorem}[Additional constant of motion] 
{\rm 
The following quantity is conserved on the solutions of \eqref{eq_A}
\begin{equation} 
I=\int \pp{\ell}{u}-A_s \pp{\ell}{A_t} \mbox{d} s \, . 
\label{int_gen} 
\end{equation} 
}
\end{theorem} 
\begin{proof}
Differentiating the quantity defined by \eqref{int_gen} with respect to time, we obtain: 
\begin{equation} 
\label{dIdt_calc}
\hspace{-2mm} 
\begin{aligned} 
\frac{d I}{d t} & = \int \partial_t \pp{\overline{\ell}_{0,A}}{u}-A_{st}  \pp{\overline{\ell}_{0,A}}{A_t} - A_{s} \partial_t  \pp{\overline{\ell}_{0,A}}{A_t} \mbox{d} s 
\\ 
& = \int - u \partial_s \pp{\overline{\ell}_{0,A}}{u} - 2 \pp{\overline{\ell}_{0,A}}{u} u_s - A p_s - A_{st}  \pp{\overline{\ell}_{0,A}}{A_t}  
\\ 
& \qquad \qquad \qquad \qquad 
- A_s \left( - \partial_s \pp{\overline{\ell}_{0,A}}{A_s} + \pp{\overline{\ell}_{0,A}}{A} +\left( p -p_{\rm ext} \right) \right) \mbox{d} s 
\\ 
&= \int - \pp{\overline{\ell}_{0,A}}{u}u_s - \pp{\overline{\ell}_{0,A}}{A_s} A_{ss} 
\\ & \qquad \qquad \qquad \qquad  - \pp{\overline{\ell}_{0,A}}{A_t} A_{ts} - 
\pp{\overline{\ell}_{0,A}}{A} A_s - \partial_s \left[ \left( p-p_{\rm ext} \right)  A\right] \mbox{d} s = 0 
\end{aligned} 
\end{equation} 
In the second line of \eqref{dIdt_calc} we have used  equations \eqref{eq_A} for time derivatives of momenta, the third line is obtained by the integration by parts and neglecting the boundary terms, and noticing that the expression under the integrand in the last line is  (minus) the full derivative of the quantity $\overline{\ell}_{0,A}+(p-p_{\rm ext}) A$ with respect to $s$. 
\end{proof}

It is hard to readily assign a physical relevance to the quantity \eqref{int_gen}, beyond it being a linear combination of the wall and fluid momenta. The expression is valid for an arbitrary Lagrangian functions, and, apart from the kinetic energy coming from the rotation and being proportional to the parameter $\lambda_*^2$, which can be set to $0$, is quite general to describe a majority of particular choices of kinetic and potential energy for models in the literature, see e.g. \cite{Pe2003,Se2016} and references within.
It was thus rather surprising to us that in spite of a large body of literature on the topic, we have failed to find any mention of the conserved quantity \eqref{int_gen}. It is probably due to the fact that the physical meaning of the constant of motion \eqref{int_gen} is not apparent, and one can hardly guess it from the equations of motion written as the balance of mass and momenta conservation using considerations from the general physical principles.

We now turn our attention to the alternative description of motion utilizing the back-to-labels map. As it turns out, this description is highly beneficial from the mathematical point of view, as it allows to combine several conservation laws into a single equation.

\subsection{Equations of motion for the back-to-labels map } 
\label{sec:single_eq}
Let us now show how to reduce the the fluid and wall momenta equations, and the conservation laws, to a single differential equation using the back-to-labels map. The procedure of this method can be explained as follows. Let us denote, for shortness, $\psi (s,t)= \varphi^{-1}(s,t)$ to be the back-to-labels map, giving the Lagrangian label of a fluid particle at the point $s$ at time $t$. For a given initial cross-section $A_0(a)$, the pressure term is entering the equations  as the Lagrange multiplier for the constraint $A(s,t)= A_0 \circ \psi \partial_s \psi = A_0(\psi) \partial_s \psi$. Instead of enforcing this holonomic constraint using the Lagrange multiplier $p$, in this Section, we will use $\psi(s,t)$ as the free variable, express all the other variables, such as the cross-section, in terms of $\psi$ in the Lagrangian, and take the variations with respect to $\psi$. It will allow us to compute the equations of motion as the explicit   Euler-Lagrange equations in terms of variables $\psi$, and not the Euler-Poincar\'e equations as we used in \eqref{full_3D}. 

We can define the function $\Psi(\alpha)$ as the antiderivative of $A_0(\alpha)$ and rewrite the conservation law as 
\begin{equation} 
\Psi(\alpha) = \int^a A_0 (\alpha') \mbox{d} \alpha' \, , \quad A= \partial_s \Psi(\psi(s,t)) 
\label{Psi_def} 
\end{equation} 
In general, $\Psi(\alpha)$ is a given function given by the initial profile of the cross-section. Since $A_0(\alpha)>0$, $\Psi(\alpha)$ is a monotonically increasing function of $\alpha$. 
We then use \eqref{Psi_def} to rewrite the conservation law  
\begin{equation} 
\label{A_cons_gen} 
A_t + \partial_s (A u) =0 \quad \Leftrightarrow \quad \Psi_{st} + \partial_s (\Psi_s u) =0 \, . 
\quad \Leftrightarrow \quad
\Psi_t + u \Psi_s =C(t) \, . 
\end{equation} 
Setting the integration constant $C(t)=0$ in the above formula, for example, due to the boundary conditions, we find the expression of $u$ in terms of the back-to-labels map $\psi$ as 
\begin{equation} 
u=-\frac{\partial_t \Psi(\psi)}{\partial_s \Psi(\psi)}=- \frac{\psi_t}{\psi_s} 
\label{formula_u} 
\end{equation} 
The Lagrangian $\ell_A$ and the corresponding action $S_A$ for an inextensible, unshearable tube with the straight centerline is written as 
\begin{equation} 
\ell_A =\int \ell_{0,A}(A, A_t, A_s, u, \lambda) \mbox{d} s \, , \quad S_A = \int \ell_A \mbox{d} t \, , 
\label{L_A_simple} 
\end{equation} 
for example, see the explicit version of $\ell_A$ given in \eqref{ell_A_explicit}. We have denoted, as before, the Lagrangian itself as $\ell_A$, which involves the integral over the $s$-domain. The  function that is being integrated, \emph{i.e.}, the integrand of the Lagrangian, is denoted with a subscript $0$, \emph{i.e.}, $\ell_{0,A}$. 

Since, according to \eqref{Q_cons}, $A=A_0(\psi) \partial_s \psi$, using \eqref{formula_u} and the constraint $A=A_0(\psi) \partial_s \psi$, we can express the Lagrangian for the tube \eqref{L_A_simple} as a function of $\psi$ only, and derive the equations of motion as a single Euler-Lagrange equation obtained by variations of the action with respect to $ \psi$. 
Let us illustrate this method in more details, including the vorticity in the tube.

In order to explicitly separate the effect of fluid's motion, let us define $b=\psi(s,t)-s$, since for the absence of fluid's motion $\psi(s,t)=s$. Then, from \eqref{formula_u}, the fluid velocity is expressed as 
\begin{equation} 
u=- \frac{b_t}{1+ b_s} \, . 
\label{u_expression_b} 
\end{equation} 
In what follows we consider, for simplicity, the initial cross-section to be constant and the tube being initially uniform, \emph{i.e}, $A_0(\alpha)=A_0=$const, corresponding to $\Psi(\alpha)=\alpha$.  In principle, we can extend the theory for an  arbitrary  
initial cross-section $A_0(\alpha) >0$, although this will introduce unnecessary algebraic complexity into equations. 

\rem{
To derive the evolution equation for $\psi(s,t)$ 
$\psi \circ \varphi (\alpha,t)= \psi(\phi(\alpha,t),t)= \alpha$ for any $t$, and substituting $\alpha=\psi(s,t)$ of both sides after differentiation, we get 
\begin{equation} 
 \big(b(s,t) + s\big) \circ \varphi(\alpha,t) = \alpha 
\, 
\Rightarrow 
\, 
b_t \circ \varphi(a,t) + (1+ b_s) \circ \varphi_t (a,t) =0  
\, 
\Rightarrow 
\, 
b_t + u (1+b_s) =0 \, . 
\label{identity_lagr_map} 
\end{equation} 
}
Instead of the full cross-section $A$, it is more convenient to work with the relative deviation of the cross-section from the equilibrium, as this approach leads to the most algebraically simple equations.   More precisely, let us define $a$ to be the relative deviation from the equilibrium value according to $A=A_0 (1+ a)$. Then, the constraint $A = A_0 \circ \varphi^{-1} \partial_s \varphi^{-1}$ becomes 
\begin{equation} 
a= b_s \, . 
\label{a_b_constr} 
\end{equation} 
\rem{ 
Since there is an explicit algebraic connection between $a$ and $R$, we can use the variable $a$ in the  Lagrangian instead of $R$. For a circular cross-section, these variables are connected through $a=R^2/R_0^2 - 1$, so the variables $a$ and $R$ are equivalent if $R \neq 0$, which is the case of pinched tube. 
} 

In general, we will have $\ell_a=\int \ell_{0,a} (a, a_t, a_s, u, \lambda) \mbox{d} s$.  Just as in \eqref{L_A_simple}, we have dropped the dependence of the Lagrangian on $a_{ss}$, the inclusion of this dependence is trivial and only leads to unnecessary algebraic complexity.  We re-define the Lagrangian using \eqref{a_b_constr}, by substituting $a=b_s$, $a_t=b_{ts}$, $a_s=b_{ss}$ into $\ell_a$, and further employing \eqref{u_expression_b} to express the velocity in terms of derivatives of $b$:  
\begin{equation} 
\ell_{0,b}=  \ell_{0,a}(a,  a_t, a_s, u, \lambda) \quad \mbox{where} \quad a=b_s \, , \quad u \mbox{ given by \eqref{u_expression_b}}. 
\label{ell_b}
\end{equation} 
We then take variations of the corresponding action defined as $S_b=\int \ell_b \mbox{d} s \mbox{d} t$
We notice that the variations of the velocity in terms of $\de b$ are given as 
\begin{equation} 
u=-\frac{b_t}{1+b_s} \, , 
\quad 
\Rightarrow 
\quad 
\de u=- \frac{\de b_t}{1+b_s} + \frac{b_t }{(1+b_s)^2} \de b_s = - \frac{1}{1+b_s} \left( \de b_t + u \de b_s \right) 
\label{u_eq_b}
\end{equation} 
and hence 
\begin{equation}
\partial_t \lambda + u \partial_s \lambda =0 \, , 
\quad 
\Rightarrow 
\quad 
\partial_t \lambda - \frac{b_t}{1+b_s} \partial_s \lambda =0 \, , 
\quad 
\Rightarrow 
\quad 
\de \lambda =\frac{ \delta b}{1+b_s} \partial_s \lambda \, . 
\label{lambda_eq_b}
\end{equation} 
Starting with the Lagrangian $\ell_a = \int \ell_{0,a}( a, a_t, a_s, u, \lambda) \mbox{d} s$, and remembering that $a=b_s$, we compute the Euler-Lagrange equations from the critical action principle 
as 
\begin{equation} 
\label{crit_action_b} 
\begin{aligned} 
\de S&= \int \left( \pp{\ell_{0,a}}{a_t} \de b_{st} + \pp{\ell_{0,a}}{a_s} \de b_{ss} + \pp{\ell_{0,a}}{a} \de b_s + 
\pp{\ell_{0,a}}{\lambda} \de \lambda + \pp{\ell_{0,a}}{u} \de u \right) \mbox{d} s \mbox{d} t
\\ 
&=\int 
 \left(  \partial^2_{st}\pp{\ell_{0,a}}{a_t}+\partial^2_{ss}\pp{\ell_a}{a_s} - 
 \partial_s \pp{\ell_{0,a}}{a} 
 \right. 
 \\ & \qquad \qquad \left. + \partial_t m_1 + \partial_s (m_1 u) +\frac{\lambda_s}{1+b_s} \pp{\ell_{0,a}}{\lambda} \right) \de b \mbox{d} s \mbox{d} t
 \\ 
 & \mbox{where} \quad m_1:=\frac{1}{1+b_s} \pp{\ell_{0,a}}{u} 
\end{aligned} 
\end{equation} 
Hence, since $\de b$ is arbitrary, the Euler-Lagrange equation corresponding to the variations with respect to $b$ gives 
\begin{equation} 
\begin{aligned} 
& \partial^2_{st}\pp{\ell_{0,a}}{a_t}+\partial^2_{ss}\pp{\ell_{0,a}}{a_s} - 
 \partial_s \pp{\ell_{0,a}}{a} + \frac{1}{1+b_s} \pp{\ell_{0,a}}{\lambda} + \partial_t m_1 + \partial_s (m_1 u)=0 \, , 
\\ 
&
 m_1=\frac{1}{1+b_s} \pp{\ell_a}{u} \, . 
 \end{aligned} 
\label{m_1_eq} 
\end{equation} 
We can simplify \eqref{m_1_eq} further and bring it to a more compact form.  Notice that due to the conservation law for volume \eqref{lambda_A_cons} for $A=A_0 (1+a)=A_0 (1+b_s)$, we have 
\begin{equation} 
\partial_t (1+b_s) +  \partial_s u (1+b_s) =0 \, , \quad \Rightarrow 
\quad 
\partial_t b+ u (1+b_s) = C(t)  \, , 
\end{equation} 
see also \eqref{A_cons_gen} for more general derivation of a non-constant initial cross-section $A_0$. 
In the expression above, we will set the integration function $C(t)=0$ assuming appropriate boundary conditions. This is the simplest choice we shall use in the remainder of the paper. 
Multiplying equation \eqref{m_1_eq} by $(1+b_s)$ and using \eqref{m_1_eq}, we can write the complete system as 
\begin{equation} 
\label{b_system_eq_0}
\begin{aligned} 
& \left( \partial_{ts} \pp{\ell_{0,a}}{a_t} + \partial_{ss} \pp{\ell_{0,a}}{a_s} - \partial_s \pp{\ell_{0,a}}{a} \right) (1+b_s) 
\\ 
& \qquad \qquad + \pp{\ell_{0,a}}{\lambda} \lambda_s + m_t + u m_s + 2 m u_s  =0 
\\ 
& a= b_s\,, \quad m:= \pp{\ell_{0,a}}{u} \, , \quad u= - \frac{b_t}{1+b_s} 
\\ 
& 
\lambda_t + u \lambda_s =0 \, . 
\end{aligned} 
\end{equation} 
We can now make approximations in the Lagrangian directly, and choosing any physically relevant Lagrangian function $\ell_{0,a}$ in \eqref{b_system_eq_0} will produce a consistent system with the Lagrangian structure. To be concrete, let us take the potential energy of the tube's wall to be $\frac{1}{2} ( P a_s^2 + Q a^2 + F(a)) $, where $F(a)$ is a smooth function describing the higher-order terms in the elastic potential, for example, $F(a)=Q_2 a^4$.  Thus, $F(a) = o(a^2)$ as $a \rightarrow 0$.   The first two terms of this potential energy are in fact the elastic part of the potential energy used in the generalized string model, see \emph{e.g.} \cite{QuTuVe2000,FoLaQu2003}, which is well-accepted in the literature. The last term $F(a)$, corresponding to the nonlinear elasticity, is one way to stabilize the dynamics in the unstable regime, as we shall see below. For further calculations, we take the Lagrangian to be 
\begin{equation} 
\label{lagr_a} 
\ell_{0,a}= \frac{1}{2} \int K a_t ^2 - P a_s^2 - Q a^2 - 2 F(a) + \rho A_0 (1+a) u^2 + \rho \Phi(a) \lambda^2  \mbox{d} s\, . 
\end{equation} 
This Lagrangian is, in essence, a particular version of \eqref{ell_A_explicit} reformulated in terms of relative deviation of the area $a$ instead of the area $A$ itself. The first term in \eqref{lagr_a} is the kinetic energy of the wall deformation; the second and third terms are the potential energy of the wall deformation. The fourth term is the kinetic energy of the fluid $\frac{1}{2} \rho A u^2$ and the last term is the kinetic energy of fluid rotation, with $\Phi(a)$ being a dimensionless quantity (shape function) of order 1. Technically, $K=K(a)$, $P=P(a)$ and $Q=Q(a)$ are functions of $a$, however, in the approximation \eqref{lagr_a} we have put them to be constants for simplicity since we assumed $a$ to be sufficiently small during the evolution. In that case, the equations \eqref{b_system_eq_0} become 
\begin{equation} 
\label{b_system_eq_explicit}
\begin{aligned} 
& \left[ K b_{ttss} - P b_{ssss} +Q b_{ss} + (F'(a))_s -\frac{1}{2} \rho  \partial_s  \left( A_0 u ^2 + \Phi'(a) \lambda^2 \right)   
\right] (1+b_s) 
\\ & \qquad + \rho \Phi (a) \lambda \lambda_s + m_t + u m_s + 2 m u_s  =0  \, ,  \\ 
&   \quad u= - \frac{b_t}{1+b_s} \, , \quad m:= \dede{\ell_a}{u} = \rho A_0 (1+b_s) u = - \rho A_0 b_t \, , 
\\
& \quad 
\lambda_t + u \lambda_s =0 \, . 
\end{aligned} 
\end{equation} 
Notice that we do not have a pressure appearing in \eqref{b_system_eq_explicit} since the constraint 
\eqref{a_b_constr} is enforced explicitly by substitution in the Lagrangian. 
The fluid momentum equation, \emph{i.e.}, the first equation of \eqref{b_system_eq_explicit}, now becomes a single equation for the quantity $b$, which is coupled with the evolution equation for $\lambda$. After some algebra, we obtain 
\rem{ 
\begin{equation} 
\begin{aligned} 
&\left( K b_{ttss} - P b_{ssss} +  (Q - \rho \Phi''(b_s)  \lambda_0^2  \right)  b_{ss} + (F'(b_s))_s  (1+b_s) 
\\
& \qquad + \rho A_0 \left( - b_{tt} -\frac{1+b_s}{2}  \partial_s \left( \frac{b_t}{1+b_s} \right)^2  +\frac{b_t}{1+b_s}  b_{st} + 2 b_t \partial_s  \frac{b_t}{1+b_s} \right) =0 \, . 
\end{aligned} 
\label{single_b_eq0} 
\end{equation} 
The nonlinear terms in \eqref{single_b_eq0} can be simplified to give, after some algebra
}
\begin{equation} 
\begin{aligned} 
&\left( K b_{ttss} - P b_{ssss} +\left( Q- \rho \Phi''(b_s) \lambda_0^2  \right)  b_{ss} + (F'(b_s))_s\right) (1+b_s) 
\\
& \qquad  \qquad + \rho A_0 \left( - b_{tt}  + \partial_s \frac{b_t^2 }{1+b_s}  \right) +
 \rho \big( \Phi(b_s) - \Phi'(b_s) (1+b_s) \big) \lambda_s  =0 \, 
 \\ 
 & \lambda_t = \frac{b_t}{1+b_s} \lambda_s \, . 
\end{aligned} 
\label{single_b_eq_lam} 
\end{equation} 
System \eqref{single_b_eq_lam} defines two coupled PDEs for the unknowns $b$ and $\lambda$. It is interesting to compare this system with the one obtained by the Euler-Poincar\'e method \eqref{inext_unshear_A}, which explicitly involves pressure, and has to be solved as a system of four equations for $(u,\lambda,A,p)$. Notice also that \eqref{inext_unshear_A} does not define an evolution equation for pressure $p$, which has to be found in order to enforce the incompressibility. In our opinion, equation \eqref{single_b_eq_lam} is advantageous for both theoretical and numerical studies, since there is no need to compute the pressure term.  
\rem{ 
Interestingly, because of the fluid momentum defined as $m:=\dede{\ell_a}{u}$ as in the second line of \eqref{b_system_eq_explicit}, there is a simple connection between the fluid momentum $m$ and fluid velocity $u$ and the quantity $b$. Indeed, the equation 
$ b_{st} + \partial_s u (1+ b_s) =0$ from the third line of \eqref{b_system_eq_explicit} integrates once to give   
\begin{equation} 
 b_{t} =-   u (1+ b_s) =0 \, , \quad \Rightarrow m= -\rho A_0 b_t \, , \quad u=-\frac{b_t}{1+b_s} \, . 
\label{m_u_b} 
\end{equation} 
In the above equation, we have set the integrating function to $0$ since we have assumed that as $s \rightarrow \pm \infty$, $b \rightarrow 0$ and $m \rightarrow 0$. For the case of fluid moving with constant speed $U$ at infinity, we  have the Lagrangian mapping $\varphi(a,t)  \rightarrow a+ U  t=s$, so $\varphi^{-1} \rightarrow s-U t$ and $b \rightarrow - U t$. Thus,   $b_t \rightarrow -U$, $u \rightarrow U$ and $b_s \rightarrow 0$ as $s \rightarrow \pm \infty$, so \eqref{m_u_b} remains valid in that case as well. 
} 
\subsection{The case of constant vorticity} 
Let us now consider the case of Section~\ref{sec:uniform_lam}, where the circulation generated by the source at the origin of the tube is held constant by a mechanical device, so $\lambda=\lambda_0$. 
The general Lagrangian with $\lambda=\lambda_*=$const in \eqref{m_1_eq} gives 
\begin{equation} 
\partial^2_{st}\pp{\overline{\ell}_{0,a}}{a_t}+\partial^2_{ss}\pp{\overline{\ell}_{0,a}}{a_s} - 
 \partial_s \pp{\overline{\ell}_{0,a}}{a}  + \partial_t m_1 + \partial_s (m_1 u)=0 \, , 
\, 
 m_1=\frac{1}{1+b_s} \pp{\overline{\ell}_{0,a}}{u} \, . 
\label{m_1_eq_0} 
\end{equation} 
where, as before, we have denoted $\overline{\ell}_{0,a}=\ell_{0,a}$ to be the Lagrangian function $\ell_{0,a}$ evaluated at $\lambda=\lambda_0$. Similarly, for a particular choice of the Lagrangian  $\overline{\ell}_{0,a}$ given by \eqref{lagr_a}, the last equation of \eqref{single_b_eq_lam} describing the evolution of $\lambda$ is satisfied identically, and the equation reduces to a single PDE for the variable $b(s,t)$.  We shall mostly analyze a particular form \eqref{single_b_eq_lam} in what follows; however, the conservation of energy we will derive will be valid for an arbitrary functions $\overline{\ell}_{0,a}$.

If we choose the relevant time and length scales $T$ and $L$ respectively, and defined the non-dimensionalized variables according to 
\begin{equation} 
\label{scaling_eqs}
\begin{aligned} 
K_*& =\frac{K}{\rho A_0 L^2} \, , \quad P_*= \frac{P T^2}{\rho A_0 L^4} \,, \quad 
Q_*:= \frac{Q T^2}{ \rho A_0 L^2} \, , 
\\
 \Phi_*(a) & = \frac{A_0}{L^2} \Phi(a) \, , \quad 
\lambda_* = \frac{T}{L^2} \lambda_0 \, , \quad 
f(b_s):=\frac{F'(b_s)}{K T^2 L^2}
\end{aligned} 
\end{equation} 
We will choose the length and time scales so that $K_*=1$ and $P_*=1$ in \eqref{scaling_eqs}, although other choices of time and length scales are also possible. 
Combining the coefficients containing the background vorticity $\lambda_*$ in a single coefficient $\zeta$ defined as 
\begin{equation} 
\zeta(a):= Q_*- \Phi_*''(a) \lambda_*^2 \, . 
\label{zeta_def} 
\end{equation} 
For example, for the cross-section approximation of the point vortex  \eqref{Phi_Vortex}, we obtain 
\begin{equation} 
\zeta(a) = Q_*+ \frac{\lambda_*^2}{2 \pi (1+a)^2} >0 
\label{P_V_approx} 
\end{equation} 
Keeping the notations $b$, $t$ and $s$ for the non-dimensionalized $b$, time and $s$, and performing some algebra, we obtain 
\rem{ 
\begin{equation} 
\begin{aligned} 
& \left( b_{ttss} -  b_{ssss} + \zeta(b_s)   b_{ss} + \partial_s f(b_s)\right) (1+b_s)  - b_{tt}
 +\frac{b_t}{1+b_s}  b_{st} + b_t \partial_s  \frac{b_t}{1+b_s} =0  \, . 
\end{aligned} 
\label{single_b_eq_non_dim2} 
\end{equation} 
The last two terms could be combined to give 
} 
\begin{equation} 
\begin{aligned} 
& \left( b_{ttss} -  b_{ssss} + \zeta(b_s) b_{ss}+ \partial_s f(b_s) \right) (1+b_s)  - b_{tt}
 +  \partial_s  \frac{b_t^2 }{1+b_s} =0  \, . 
\end{aligned} 
\label{single_b_eq_non_dim} 
\end{equation} 
Equation \eqref{single_b_eq_non_dim} gives, as far as we know, a new result expressing the evolution of the collapsible tube with compliant walls carrying vorticity in terms of back-to-labels mab $b(s,t)$. While it is difficult to observe the variable $b$ in an experiment, its spatial derivative $a=b_s$ can readily be observed. The rest of the paper is dedicated to the analysis of \eqref{single_b_eq_non_dim}, both in general case, and in the asymptotic approximations of long wavelengths and large times. 

We start with the analysis of the full equation \eqref{single_b_eq_non_dim} and its general counterpart \eqref{m_1_eq} first. 

\subsection{Energy conservation and additional constant of motion} 
\paragraph{Energy conservation} 
Let us define $\ell_{0,b}$ to be $\overline{\ell}_{0,a}$, \emph{i.e.}, the Lagrangian function $\overline{\ell}_{0,a}$ with the substitution $a=b_s$, $a_t=b_{st}$, $a_s=b_{ss}$. Since \eqref{m_1_eq_0} is derived as an Euler-Lagrange equation for the Lagrangian 
\begin{equation} 
\ell_b=\int \ell_{0,b} (b_{st}, b_t, b_s, b_{ss}) \mbox{d} s\, , 
\label{ell_b_def}
\end{equation} 
it conserves the energy-like function as stated in the following 
\begin{lemma}[Energy conservation] 
\label{lemma:energy} 
{\rm 
For appropriate periodic boundary conditions, the equation \eqref{single_b_eq_non_dim} conserves the quantity 
\begin{equation} 
E=\int e \mbox{d}s \, , \quad  e:= b_{st} \pp{\ell_b}{b_{st}} + b_t \pp{\ell_b}{b_t} - \ell_b  \, .  
\label{E_def} 
\end{equation} 
} 
\end{lemma} 
\begin{proof} 
We prove this theorem by computing the evolution of energy density $e$ defined by \eqref{E_def}. 
\begin{equation} 
\label{e_cons_deriv}
\begin{aligned} 
\partial_t e & = b_{stt} \pp{\ell_b}{b_{st}} + b_{tt} \pp{\ell_b}{b_t} + 
b_{st} \partial_t \left( \pp{\ell_b}{b_{st}} \right) + b_t \partial_t \left( \pp{\ell_b}{b_{t}} \right) - \frac{d\ell_b}{dt}  
\\ & = b_t \left[- \partial_{st} \pp{\ell_b}{b_{st}} +\partial_t \pp{\ell_b}{b_t} +  \partial_s \pp{\ell_b}{b_s}  
- \partial_{ss} \pp{\ell_b}{b_{ss}} \right] 
\\ & \qquad \qquad \qquad \qquad + 
\partial_s \left( b_t \partial_t \pp{\ell_b}{b_{st}} -b_t  \pp{\ell_b}{b_{s}} -b_{st}  \pp{\ell_b}{b_{ss}} \right) 
\end{aligned} 
\end{equation} 
The expression in the square brackets is exactly the Euler-Lagrange equation for the Lagrangian \eqref{ell_b_def}, so it vanishes identically. Thus, we have $\partial_t e + \partial_s J_e=0$ where the energy flux $J_e$ is defined by the equation in the parenthesis on the last line of \eqref{e_cons_deriv}. Thus, $\int e_t \mbox{d} s$ vanishes for periodic boundary conditions of the conditions where $J_e=0$ on the boundaries, for example, $s \in (- \infty, \infty)$ with all variables tending to $0$ as $s \rightarrow \pm \infty$. 
\end{proof} 

In the case when at the undisturbed state the fluid is moving with a constant velocity $U$, it is is convenient to write \eqref{single_b_eq_non_dim} in the form $b(s,t) = - U t + B(s,t)$ so $B(s,t)=0$ if there are no deviations of the flow from the steady flow $u=U$. Then, \eqref{single_b_eq_non_dim} becomes 
\begin{equation} 
 \left( B_{ttss} -  B_{ssss} + \zeta(B_s) B_{ss}+ \partial_s f(B_s) \right) (1+B_s)  - B_{tt}
 +  \partial_s  \frac{(B_t-U)^2 }{1+B_s} =0  \, . 
\label{eq_B} 
\end{equation} 
In addition to the energy conservation, one can prove the appropriate equivalent of \eqref{int_gen} reformulated in terms of the back-to-labels map $B$. Namely, we have the following 
\begin{lemma}[Additional constant of motion] 
{\rm For periodic boundary conditions, or boundary conditions where $B \rightarrow 0$ when $s \rightarrow \pm \infty$, equation \eqref{eq_B} conserves the following quantity: 
\begin{equation} 
I= \int B_t + B_{ss} B_{t s} \mbox{d} s \, . 
\label{Cons_law}
\end{equation} 
}
\end{lemma} 
\begin{proof} 
This can be seen by re-writing \eqref{eq_B} in the form $\partial_t I+ \partial_s J$ for some quantity $J$ depending on $B$, $B_t$, $B_s$, $B_{ss}$ {\em etc.}. The result is valid for arbitrary functions $\zeta(a)$ and $f(a)$, and arbitrary velocity  $U$. 
\end{proof} 

\subsection{Linear stability and instability} 
Let us now turn our attention to the linear stability analysis of equation \eqref{single_b_eq_non_dim}. 
Since we assume that $f(b_s) = o (b_s)$, we can write the linearized version of that equation, by assuming small amplitude of $b$ with moderate values for $t$- and $s$-derivatives. For convenience, we also define $\zeta_0 := \zeta(0)$ to be the linearized value of $\zeta$ at $b=0$. Then, \eqref{single_b_eq_non_dim} gives 
\begin{equation} 
\label{lin_b_eq} 
\left( b-b_{ss}\right)_{tt} = \left( \zeta_0 b - b_{ss} \right)_{ss}  \, . 
\end{equation}
Note that on the left of that equation, we obtain exactly the Helmholtz operator acting on $b$. On the right of that equation, we have a Helmholtz operator for $\zeta_0>0$ and anti-Helmholtz operator for $\zeta_0<0$. We will consider cases of $\zeta_0$ being possibly positive or negative, keeping in mind that for physical reasons, $\zeta_0>0$ due to the expression \eqref{P_V_approx}. While the systems values $\zeta_0<0$ may exist, their physical origin still remains uncertain, and thus all our numerical solutions will be computed for $\zeta_0>0$ \eqref{P_V_approx} unless explicitly stated. 
 
To compute the dispersion relation for no external flow, assume $b= b_0 e^{i (k s - \omega t)}$ and substitute into \eqref{lin_b_eq}. We obtain:
\begin{equation} 
\omega^2 (1+k^2) = k^2 ( \zeta_0 + k^2) \, , \quad \omega(k) = \pm k \sqrt{ \frac{\zeta_0+k^2 }{1+k^2} } \, . 
\label{disp_rel} 
\end{equation} 
For large $k$, the waves propagate with the constant speed (phase velocity) of $\pm 1$. If $\zeta \geq 0$, all wavelengths are stable. If $\zeta_0<0$, then there is a band of instability $0<k<\sqrt{- \zeta}$ where the waves in that band grow exponentially. Long wavelengths $k \rightarrow 0$ propagate with the phase and group velocities $\sqrt{\zeta_0}$ and short wavelengths 
$k \rightarrow \infty$ propagate with the speed $1$. 

Let us now compute the stability for the case when the base flow has the circulation $\lambda_*$ and velocity $U$ in the dimensionless variables.  The linearization of the terms in the equation \eqref{eq_B} proceeds as follows: 
\begin{equation} 
b=- U t + \epsilon B\, , \quad b_t = -U + \epsilon B_t \, , \quad b_s = B_s \, , \quad \epsilon \ll 1 \, . 
\label{b_lin} 
\end{equation} 
Then, the linearization of \eqref{single_b_eq_non_dim} up to order $\epsilon$ is given by 
\begin{equation} 
\label{b_eq_lin_gen} 
B_{ttss}-B_{ssss}+ (\zeta_0-U^2)  B_{ss} - B_{tt} - 2 U B_{st}=0 \, . 
\end{equation} 
Substitution of $B=e^{i k x-i\omega t}$ gives the dispersion relation connecting $\omega$ and $k$ 
\begin{equation} 
\omega^2 (1+k^2) + 2 U k \omega - k^2 (\zeta_0 -U^2 + k^2) =0 \, . 
\label{disp_rel_eq} 
\end{equation} 
We can thus solve for $\omega=\omega(k)$ as 
\begin{equation} 
\omega= \frac{k}{1+k^2} \left( - U  \pm \sqrt{ U^2 +(1+k^2) (\zeta_0 -U^2+ k^2) } \right) \, . 
\label{disp_rel_U}
\end{equation} 
We want to find the areas of stability and instability, \emph{i.e.}, the values of parameters $U$ and $\zeta_0$ such that ${\rm Im} \omega=0$ (stability) or ${\rm Im} \omega >0$ (instability). Because of the $\pm$ sign in \eqref{disp_rel_U}, it is sufficient to find the values of parameters when $\omega$ becomes complex for instability. 
The square root (discriminant) in \eqref{disp_rel_U} is non-negative if and only if the quadratic polynomial $f(z)=-U^2 z+(1+z)(\zeta_0+z)$ is non-negative for all $z=k^2 \geq 0$. 
Therefore, the system is unstable if and only if 
\begin{equation} 
U^2 \geq {\rm min}_k \frac{(1+k^2)(\zeta_0+k^2)}{k^2} 
\label{u_cond} 
\end{equation} 
Since the minimum of the function on the right-hand side of \eqref{u_cond} is achieved at $k^2 = \sqrt{\zeta_0}$ (we remind that $\zeta_0$ is assumed to be positive), then the condition of instability is given as 
\begin{equation} 
|U| > U_*=1+ \sqrt{\zeta_0}. 
\label{instab_cond_2} 
\end{equation} 
For values of $U$ satisfying \eqref{instab_cond_2}, there is a band of unstable wavelengths of $k$ satisfying \eqref{u_cond}. When $U$ is slightly above $U_*$, the instability band of $k$ is centered around $k_*=\zeta_0$. Thus, at the bifurcation point $U=U_*$, the typical unstable wavelength is $k \sim \sqrt{\zeta_0}$. 

We shall also note that for for the long wavelength $k \rightarrow 0$, equation \eqref{disp_rel_eq} gives $\omega \sim v_{\pm} k$, where the wave speed $v_{\pm} = -U \pm \sqrt{\zeta_0}$. The long wavelength are thus always stable for $\zeta_0>0$. In the short wave limit $k \rightarrow \infty$, equation \eqref{disp_rel_eq} gives $\omega = \sim \pm k$, so the short waves propagate with speeds of $\pm 1$. 

The results of linear stability analysis presented above are summarized on Figure~\ref{fig:stability}. We take $\zeta_0=1$ corresponding to the critical velocity $U_*=2$. The top and the bottom panels of this figure show, respectively, the real part of $\omega(k)$ computed from \eqref{disp_rel_U}. On the left panels, we show the results with $U=1$, which is the stable case. For the stable case, ${\rm Im} \omega(k)=0$ for all $k$. The right panels show the results for $U=3$ which is the unstable case. There is a band of instability centered around $k \sim \zeta_0 =1$ as expected from the discussion after \eqref{u_cond}. 

\begin{figure}[h]
\centering
\includegraphics[width=0.48\textwidth]{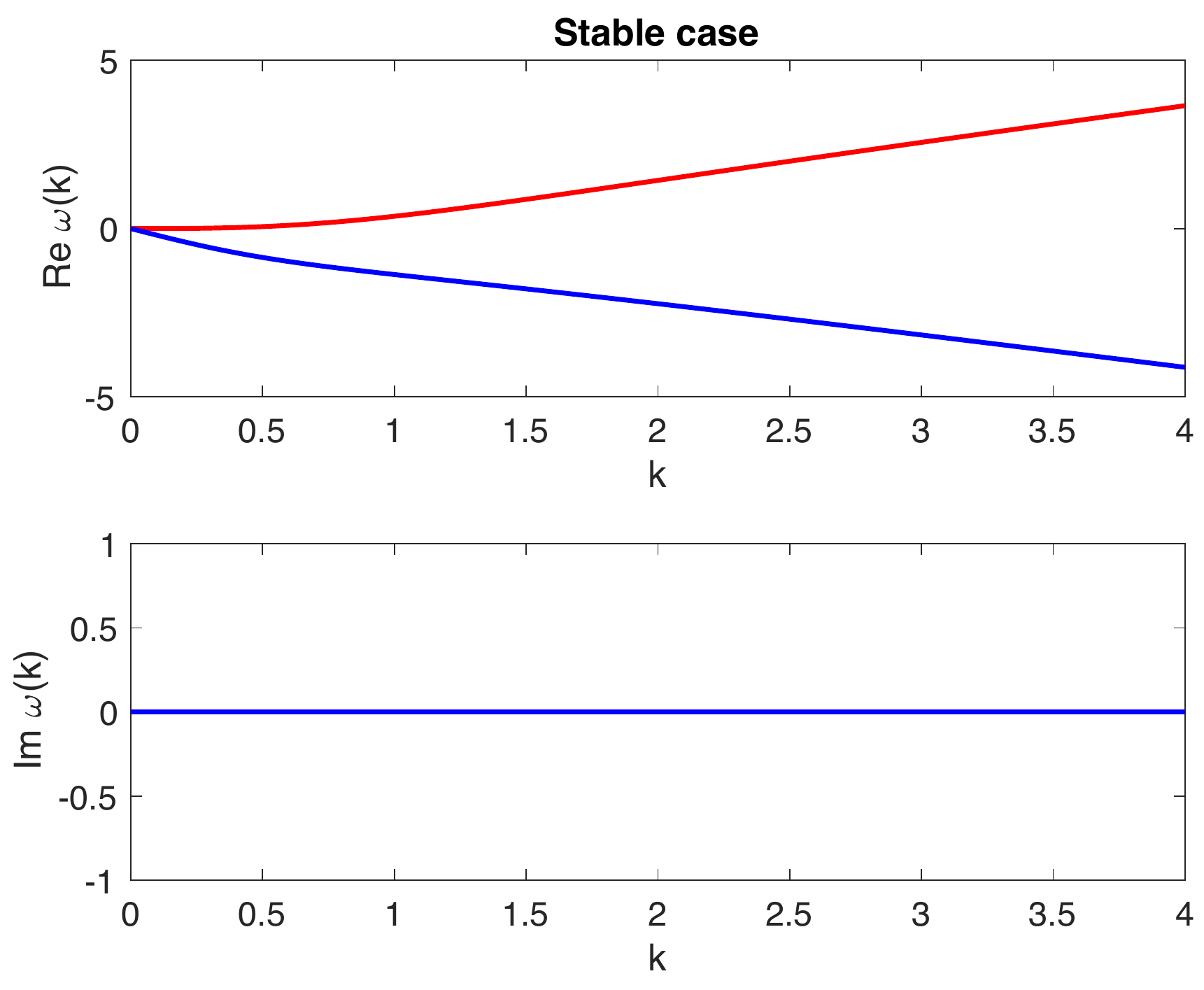}
\includegraphics[width=0.48\textwidth]{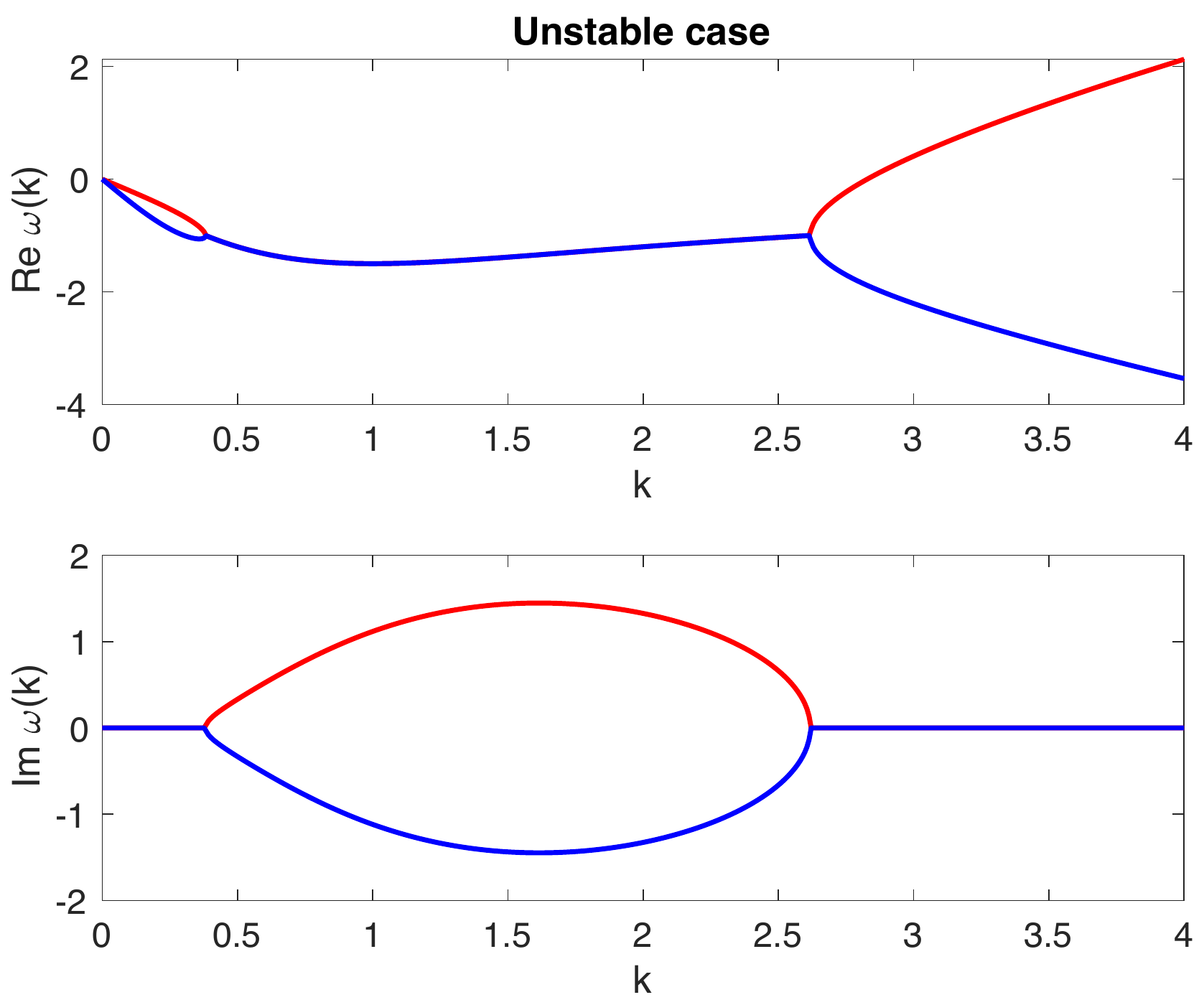}
\caption{Examples of  $\omega(k)$ computed from \eqref{disp_rel_U}, with $\zeta_0=1$. Left panels: $U=1$ (stable case). Right panels: $\zeta_0=0$ (unstable case). Solid red and blue line correspond to the respective choices of the $+$ and $-$ sign in \eqref{disp_rel_U}.}
\label{fig:stability}
\end{figure}

\subsection{Numerical simulations: stable case} 
\rem{ 
For numerical solution, it is convenient to extend \eqref{lagr_a} to avoid potential singularity at $a=-1$, and include the curvature terms proportional to the curvature $(1/R-1/R_0)$. In \eqref{lagr_a}, these terms were approximated by the term $Q a^2$. A more physically relevant expression is 
\begin{equation} 
\label{lagr_a_2} 
\begin{aligned}
\ell_a &= \frac{1}{2} \int K a_t ^2 - P a_s^2 - Q a^2 -F(a) + \rho A_0 (1+a) u^2 + \rho \Phi(a) \lambda^2 \, , 
\end{aligned} 
\end{equation} 
} 
The results of simulations of \eqref{eq_B} with periodic boundary conditions for $a(s,t)=B_s(s,t)$ are presented on Figure~\ref{fig:simulations}. We have taken $Q_*=1$, and $\zeta(a)$ given by \eqref{P_V_approx}. For the stable case, we have also set the nonlinear terms to $0$, \emph{i.e.}, $f(a)=0$. The parameters of the simulations on the Figure are $U=1$ and $\lambda_*=1$, with $\zeta_0 \simeq 1.1592$, corresponding to the stable case.  We present $a=B_s$ since $a(s,t)$ is the physical observable, being the deviation of the relative area from its equilibrium value. 

 On the left panel we present $a=B_s$ inferred from the solution $B(s,t)$ of the full equation \eqref{eq_B} starting with a localized disturbance $B(s,0)=e^{-0.1 s^2}$, $B_t(s,0)=0$.  On the right panel, we present $a=B_s$ inferred from the solution for  the Boussinesq-like approximation of the full equation \eqref{app2}, derived below in Section~\ref{sec:Boussinesq} with the same parameter values and initial conditions as on the left panel. 
\begin{figure}[h]
\centering
\includegraphics[width=0.48\textwidth]{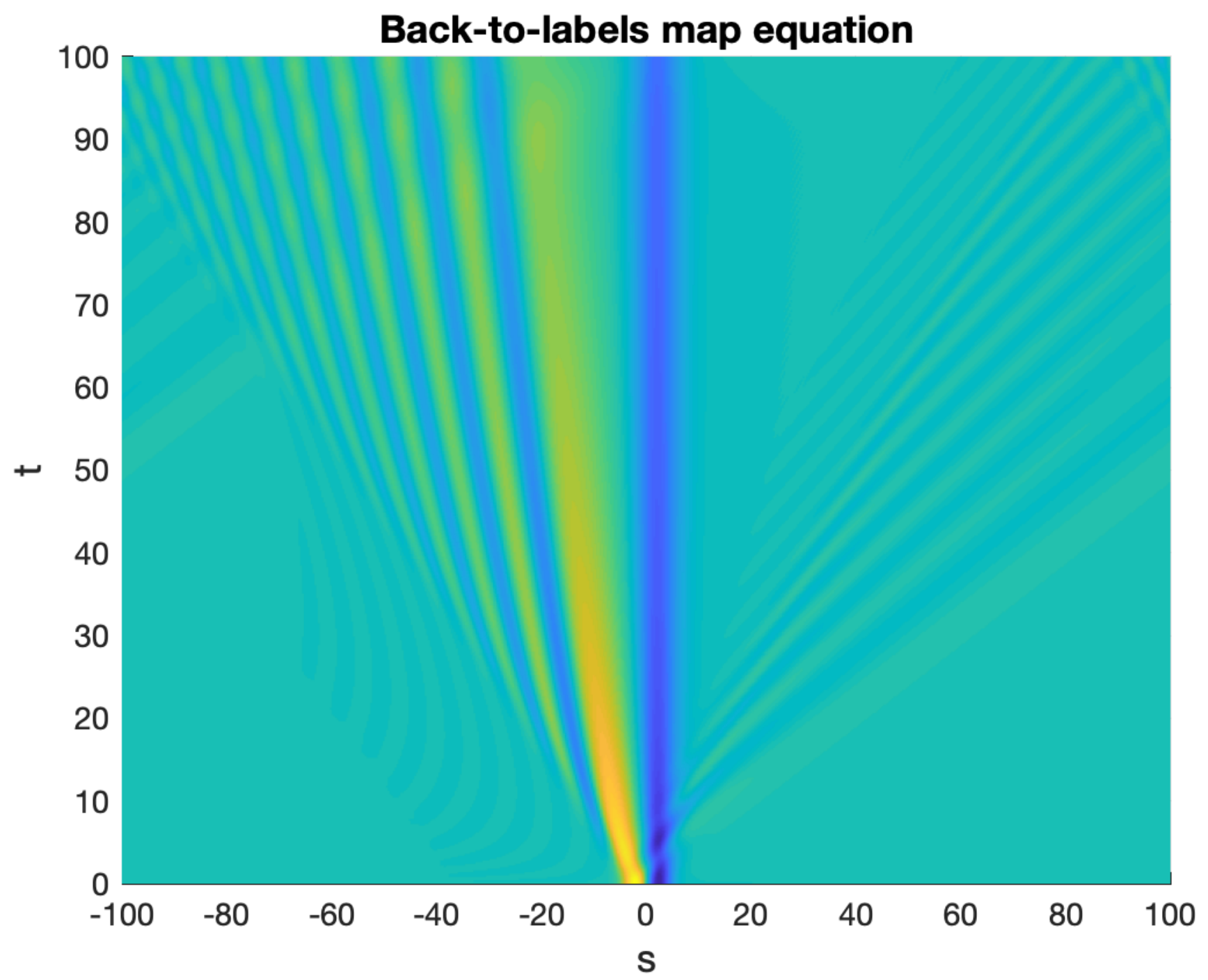}
\includegraphics[width=0.48\textwidth]{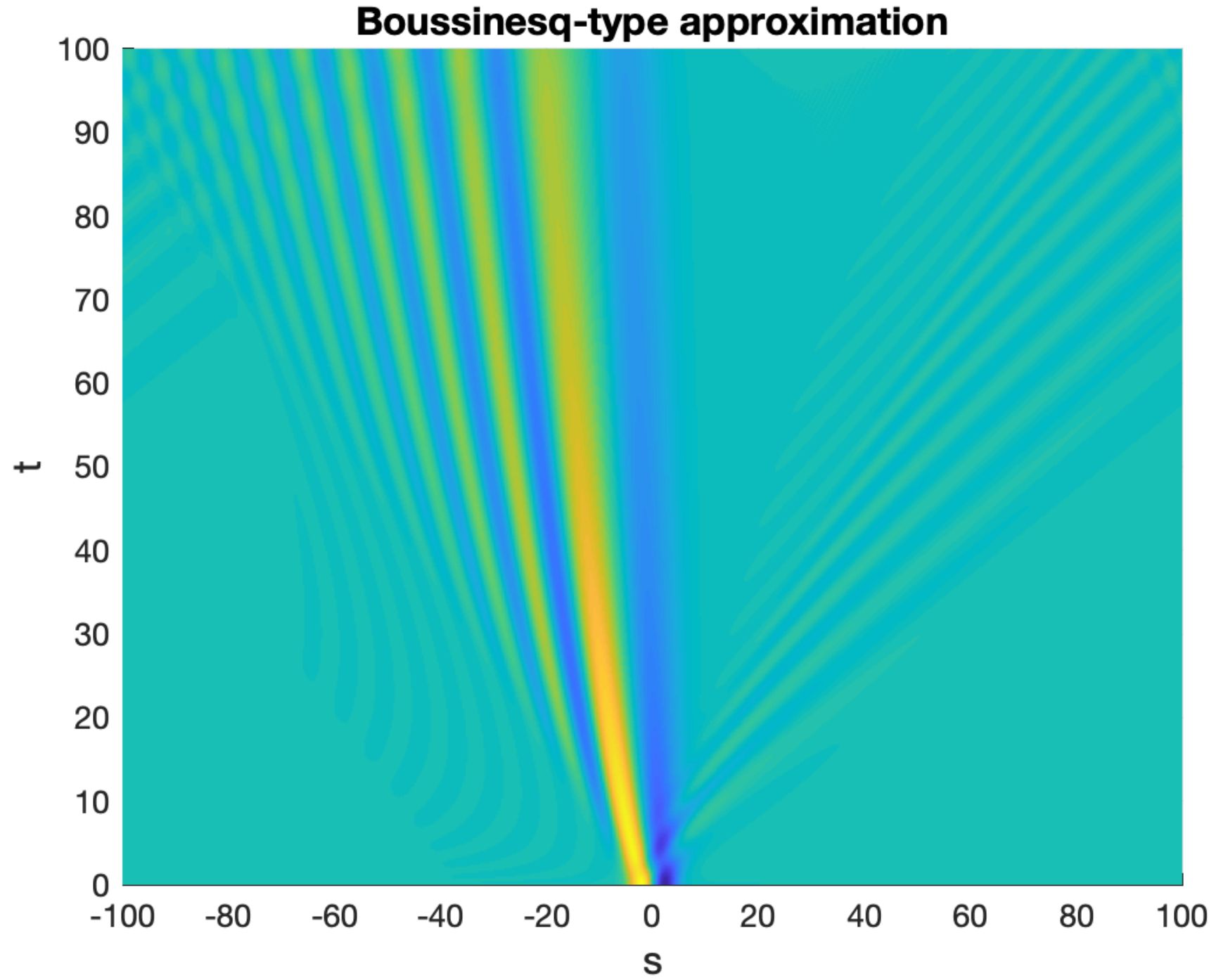}
\caption{Examples of solutions for $a=B_s(s,t)$. Left: full equation \eqref{eq_B}. Right: Boussinesq-like approximation \eqref{app2} (see below).    }
\label{fig:simulations}
\end{figure}
To check the accuracy of our calculations, on Figure~\ref{fig:Cons_Law}, we also show the energy defined by \eqref{E_a} (left panel) and the conserved quantity $I$ defined by \eqref{Cons_law} (right panel). The energy defined in \eqref{E_a}  is preserved to the relative accuracy of $10^{-4}$.  For comparison, on the right panel we also present the terms $B_t$ and $B_{ss} B_{ts}$, (green dotted and blue dashed lines) which are of order one. The constant of motion defined by  \eqref{Cons_law}, shown as a solid red line, is preserved to expected precision, about $10^{-4}$. 
\begin{figure}[h]
\centering
\includegraphics[width=0.48\textwidth]{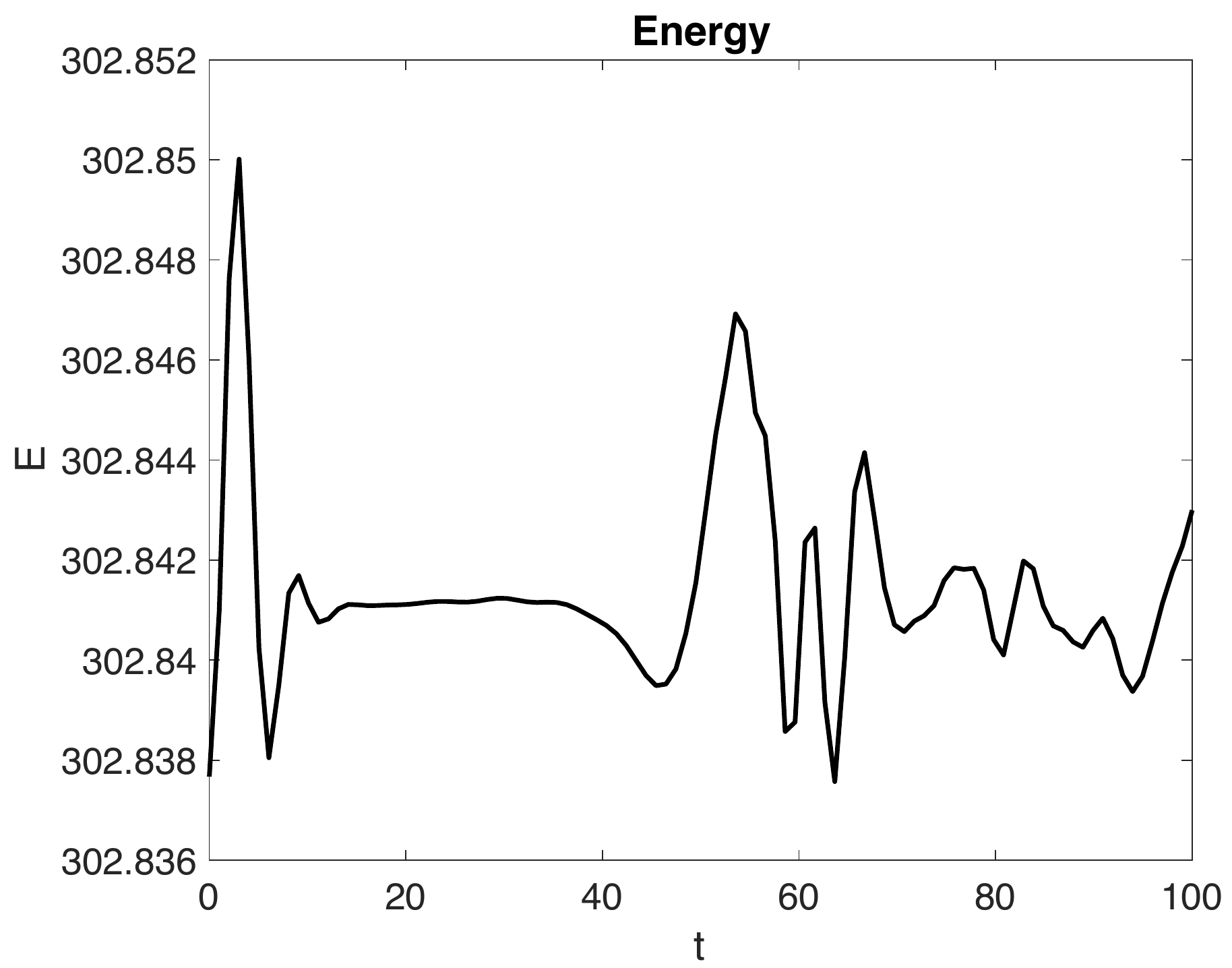}
\includegraphics[width=0.48\textwidth]{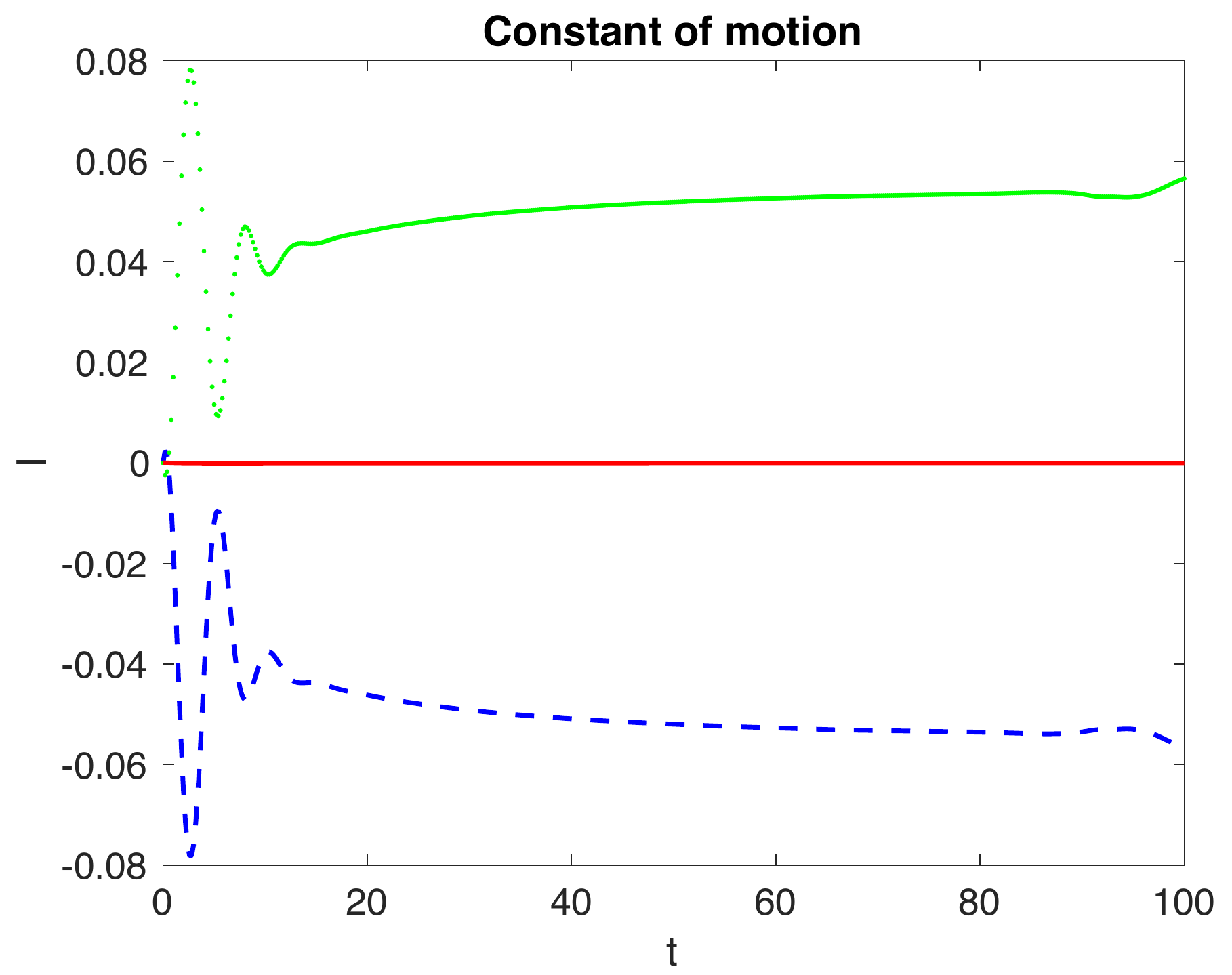}
\caption{Computed value of the conserved quantities for simulation shown on Figure~\ref{fig:simulations}. Left panel: energy $E$ as defined in \eqref{E_def}. Right panel: quantity defined by \eqref{Cons_law} Dashed blue line: $\int B_t \mbox{d} s$. Dotted green line:  $\int B_{ss} B_{ts} \mbox{d} s$. 
Solid red line: the conserved quantity $I=\int B_t + B_{ss} B_{ts} \mbox{d} s$.  }
\label{fig:Cons_Law}
\end{figure}

If $|U|> 1+ \sqrt(\zeta_0)$, the solution becomes unstable, and if there is no potential energy to limit the growth of the solution, it quickly reaches the singular (collapsed) state $a=-1$.  In reality, any physical system will have a potential energy preventing the singularity. One method is to introduce higher derivatives in the elasticity function $F$, and another way is to consider the potential energy preventing both indefinite growth of $a$ and collapse to $a=-1$. We chose the same parameters as in \eqref{fig:simulations} except  
\begin{equation} 
\label{param_unstable} 
U=1.25 + \sqrt{\zeta_0}, \quad   \mbox{ and  } \quad 
f(a) = \frac{d}{d a} \frac{a^4}{(1+a)^2} =  \frac{2 a^3( 2 +a) }{(1+a)^3} 
\end{equation} 
The initial conditions at $t=0$ are taken to be a random perturbation for $B(s,0)$ with the amplitude $10^{-3}$ and $B_t(s,0)=0$. The results of the simulations for $a=B_s$ are demonstrated on Figure~\ref{fig:simulations_unstable}. The solution grows exponentially from the unstable equilibrium state $a=0$, saturating at an amplitude $|a| \sim 0.5$ and yielding complex non-steady dynamics. 
\begin{figure}[h]
\centering
\includegraphics[width=0.8\textwidth]{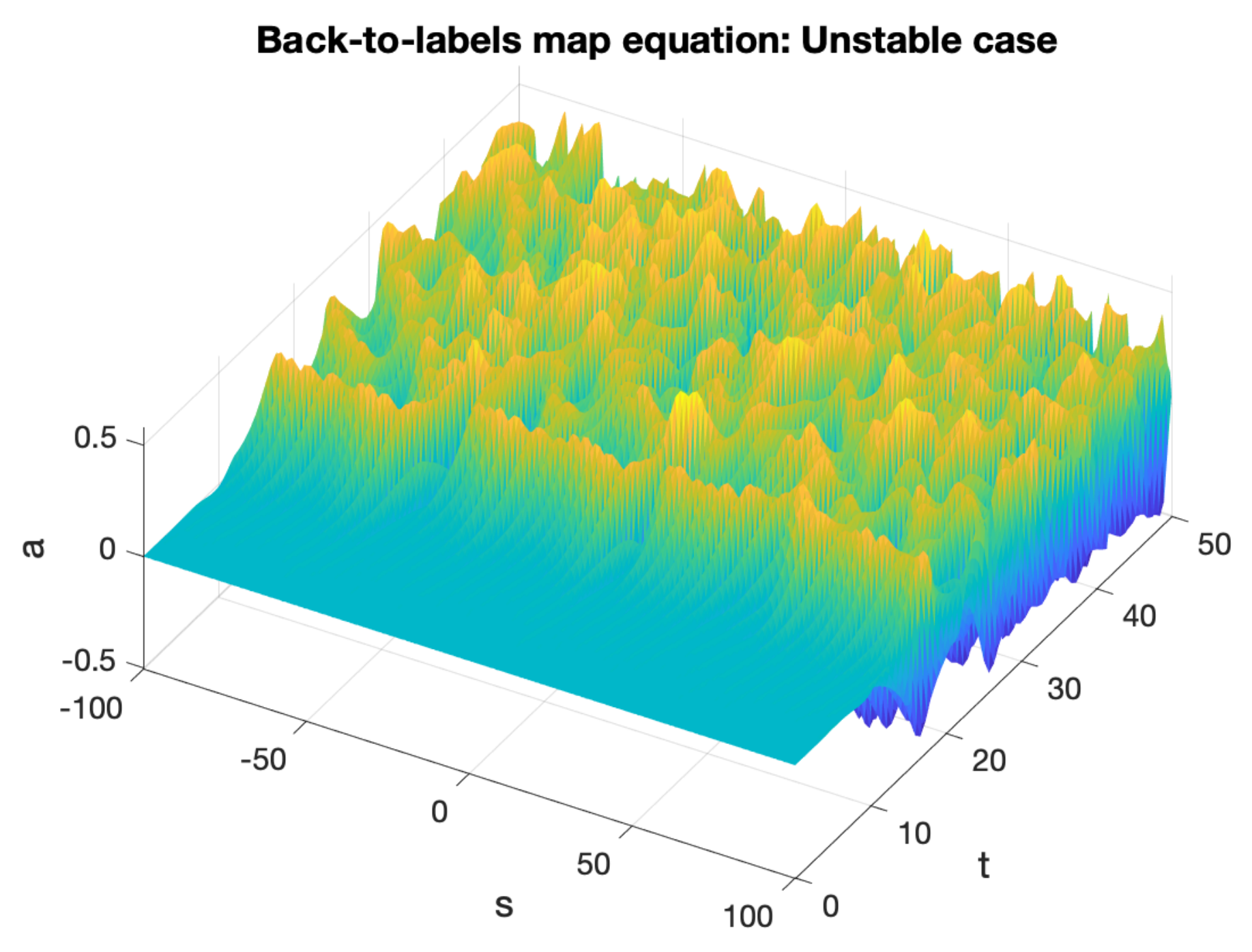}
\caption{ Solution $a=B_s$ for the unstable case $U>1+ \sqrt{\zeta_0}$. Initial conditions are given by small random perturbations of order $10^{-3}$ to the equilibrium state $B=0$.  }
\label{fig:simulations_unstable}
\end{figure}

\section{Asymptotic analysis of equation \eqref{eq_B}} 
\label{sec:approximate_sols}
\subsection{Small cross section (collapse) and reduction to Monge-Amp\`ere equation}
One of the most interesting regimes of evolution occurs when the cross-section almost collapses to 0, in other words 
$ 1+B_s>0$ and $ |1+B_s| \ll 1$. This regime is important for physiological flow applications and is studied quite extensively in the literature, see, for example, recent review \cite{Han-etal-2013} and the references therein. 
We show that in this regime, our equation \eqref{eq_B} is approximated by the celebrated Monge-Amp\`ere equation. 
We proceed as follows. 
For $1+B_s$ being small and positive, and $B_{ss}$, $B_{sstt}$, $B_{ssss}$ remaining finite, the asymptotic equation following from \eqref{eq_B} is 
\begin{equation} 
\label{B_asymptotic} 
B_{tt}=\partial_s\frac{(B_t - U)^2}{1+B_s}.
\end{equation} 
We now use the substitution 
\begin{equation} 
B=K_s(s,t) -s +Ut.
\label{B_subst}
\end{equation}  With this substitution, \eqref{B_asymptotic} integrates once with respect to $s$. Dropping the arbitrary function of time appearing under this integration, we obtain 
the following equation for $K(s,t)$: 
\begin{equation}  
\label{K_eq} 
K_{ss}K_{tt}=K_{st}^2\, , 
\end{equation} 
which is the (real) homogeneous Monge-Amp\`ere equation with $(s,t)$ being independent variables. For the approximation \eqref{B_asymptotic} to be valid, we assume $K$ to be order $1$, and since $B_s+1$ is small, we need that $K_{ss}$ should be small. We can thus assume that $K=K(\epsilon S, \epsilon \tau)$ where $\epsilon \ll 1$ is a small parameter and the variables $S, \tau$ are of order 1. Then the equation \eqref{K_eq} transforms to 
\begin{equation}  
\label{K_eq_2} 
K_{SS}K_{\tau \tau }=K_{S \tau}^2\, , 
\end{equation} 
\emph{i.e.} satisfies the same Monge-Amp\`ere equation. Thus, any solution $K(S, \tau)$ of \eqref{K_eq_2} will generate a solution $K(s,t)$ describing the collapse of the tube. For physical reasons, we need $K_{SS}>0$ for all solutions since $1+B_s = K_{ss}>0$. For that solution, the neglected terms in \eqref{eq_B} are order $\epsilon^3$ or smaller in order than the terms kept in \eqref{B_asymptotic}, so the approximation is consistent. 

As is known from the literature, equation \eqref{K_eq_2} possesses a rich structure of solutions and is integrable \cite{Au1982,Ca1991,GuBr2001}. 
For example, for any smooth function $g(x)$ with $g''(x)>0$,  $K(S, \tau)=g(S-c \tau)$  is a solution of \eqref{K_eq_2} for any constant $c$. Thus, the collapse of the tube is strongly dependent on the initial conditions and is unlikely to be described by a universal law. 

\rem{ 
\section{Large Cross section regime}

In this regime the area $1+B_s \gg 1$ is large (e.g. growing with $t$ exponentially due to instability). Then suppose $\mathcal{O}(B_s)=\mathcal{O}(B_t)\gg 1$ The equation can be written as

$$ B_{ttss}-B_{ssss} + \zeta B_{ss}  -\frac{ B_{tt}}{1+B_s}+ \frac{ 1}{1+B_s}\left( \frac{(-U+B_t)^2}{1+B_s} \right)_s =0$$ 

Then the behavior is determined basically from

$$ B_{ttss}-B_{ssss} + \zeta B_{ss} =0 $$

The dispersion is $$c(k)=\pm\sqrt{1+\frac{\zeta}{k^2}}$$ and if $ \zeta >0 $ all waves propagate in a stable pattern. 

However, if $ \zeta <0 $ then for $0<k^2 < |\zeta|$ there is an instability caused by waves of the type 
$$e^{\pm \sqrt{ |\zeta|-k^2} t}\sin(ks)$$ with exponentially growing amplitudes and this scenario seems to be possible? 
 } 

\subsection{Approximation by Boussinesq-type equations} 
\label{sec:Boussinesq}
In this section, we are interested in a consistent approximation of \eqref{eq_B} for small amplitude and long wavelength approximation which is a common approximation for hydrodynamic system with moving boundaries. Perhaps the most celebrated examples are the Korteveg-de Vries (KdV) and  Boussinesq equations. We start with the reduction of our system to the Boussinesq-type equation. Such reduced models have been used in the analysis of pulse propagation in arteries, see \cite{Ca2012,LiLiBo2014}. 

In order to keep the description self-contained, let us briefly introduce the integrable Boussinesq equation (BE), 
\begin{equation} \label{BE}
u_{tt}=u_{xx}+3(u^2)_{xx}+u_{xxxx}.
\end{equation}
Note that the notations $u, \psi, \lambda$ in this section are not related to variables from the previous sections. We hope no confusion arises from the clash of notations, since these are the variables used commonly in the field of integrable equations. 

Equation \eqref{BE} arises in many various physical applications, most notably in describing propagation of the long waves in shallow water \cite{B1871,B1872}. The BE is integrable by the inverse scattering method, with the Lax pair representation that originates from Zakharov \cite{Z73}:
\begin{align}
4\psi_{xxx}&+ (6 u+1) \psi_x + 3u_x \psi \pm \sqrt{3} i (\partial_x^{-1}u_t)\psi = \lambda \psi ,\nonumber \\
\psi_t&=\mp i \sqrt{3} (\psi_{xx }  +u \psi). \nonumber \\
                   \label{Lax}
\end{align}
Here $\psi$ is the eigenfunction of the spectral problem and $\lambda$ is the spectral parameter. 
 The  solutions are explicitly described in terms of solitons which are computed by Hirota \cite{H73}. 
It is more convenient to rescale the variables in \eqref{BE} 
\begin{equation} 
 u \rightarrow \alpha u, \qquad \partial_t \rightarrow \beta \partial_t, \qquad \partial_x \rightarrow \beta \partial_x 
 \label{rescale_def} 
 \end{equation} 
 where $\alpha$ and $\beta$ are some constants.  The BE equation \eqref{BE} then becomes 
\begin{equation} \label{BE1}
u_{tt}=u_{xx}+3\alpha(u^2)_{xx}+\beta^2 u_{xxxx}.
\end{equation}

Let us now turn our attention to the equation \eqref{eq_B} for $\zeta=$const and no nonlinear terms in the potential, \emph{i.e} $f(a) =0$. We get the following nonlinear equation, written here for reference: 
\begin{equation} 
\label{eq_B_0} 
(B_{ttss}-B_{ssss} + \zeta B_{ss})(1+B_s)- B_{tt}+ \left( \frac{(-U+B_t)^2}{1+B_s} \right)_s =0 \, .
\end{equation} 
This equation will form the basis of derivation of approximate limiting cases in this chapter. 

Let us consider a small amplitude and large wavelength propagation regime, where $B\rightarrow \varepsilon B ,$  $ \partial_t \rightarrow \varepsilon \partial_t,$ $ \partial_x \rightarrow \varepsilon \partial_x $ where $\varepsilon$ is a small parameter that defines the scale for the corresponding quantities. 
By keeping terms up to and including  $\varepsilon ^2$ in \eqref{eq_B_0} we obtain 
\begin{equation} \label{app1}
 (\zeta-U^2) B_{ss} -2UB_{ts} - B_{tt}+ \varepsilon^2 (B_{ttss}-B_{ssss}) + \varepsilon^2 \left( \frac{\zeta}{2} B_s^2 + B_t^2 + 2U B_t B_s  \right)_s =0 .
\end{equation}
The leading order terms in this equation are
\begin{equation} 
 (\zeta-U^2) B_{ss} -2UB_{ts} - B_{tt} =0 
\end{equation}
giving an operator equality to the leading order 
\begin{equation} 
\label{t_deriv_approx}
 \partial_t = (\pm \sqrt{\zeta}-U)\partial_s + \mathcal{O} (\varepsilon) = - c_{\pm} \partial_s + \mathcal{O} (\varepsilon) \, , 
 \quad c_{\pm} := U \pm \sqrt{\zeta} 
 \end{equation}
when the operators are acting on the solution $B$ of \eqref{app1} .  The plus sign is for the right - running waves, the minus sign is for the left running waves. Then, $c_+$ is, to the leading order, the propagation speed of the right-running waves and $c_-$ is the speed of the left-running waves. 

Next, we use the expression for $\partial_t$ from \eqref{t_deriv_approx} to make an approximation of the last term in \eqref{app1} 
\begin{align}
   \label{quad}
   &\frac{\zeta}{2} B_s^2 + B_t^2 + 2U B_t B_s = \frac{\zeta}{2} B_s^2 + (\zeta \mp 2\sqrt{\zeta}U + U^2)B_s^2 + 2U (\pm\sqrt{\zeta}-U)B_s^2 + \mathcal{O}(\varepsilon) \notag \\
   &=\left(\frac{3\zeta}{2}-U^2 \right) B_s^2 + \mathcal{O}(\varepsilon) \notag 
\end{align}
Then, this approximation of \eqref{app1} gives an equation simultaneously valid for both left- and right-running waves:
\begin{equation} \label{app2}
 (\zeta-U^2) B_{ss} -2UB_{ts} - B_{tt}+ \varepsilon^2 (B_{ttss}-B_{ssss}) + \varepsilon^2 \left( \frac{3\zeta}{2}-U^2\right ) (B_s^2)_s =0. 
\end{equation}
This equation provides a universal, Boussinesq-type approximation of equation \eqref{eq_B}. We present a numerical solution of this equation on the right panel of the Figure~\ref{fig:simulations} to compare with the solution of full equation \eqref{eq_B} with the same initial conditions. We observe that there is a very good quantitative agreement between these solutions for this propagation regime. Note that \eqref{app2} is substantially algebraically simpler than \eqref{eq_B}. We do not yet know whether the equation \eqref{app2} is integrable, but from the algebraic structure of this equation, we believe that it is much more amenable to analytical studies than the full \eqref{eq_B}.

While it is not known whether \eqref{app2} is integrable, a further approximation can be made to reduce it to an integrable Boussinesq case. We can make the following approximation using \eqref{t_deriv_approx}: 
\begin{equation} 
\label{4th_deriv_approx} 
B_{ttss}-B_{ssss}=[( \sqrt{\zeta}\pm U )^2 -1]B_{ssss} +\mathcal{O}(\varepsilon) \, . 
\end{equation} 
This approximation,  substituted in  \eqref{app2} gives the following equation: 
\begin{equation} \label{app3}
 (\zeta-U^2) B_{ss} -2UB_{ts} - B_{tt}+ \varepsilon^2 [( \sqrt{\zeta}\pm U )^2 -1]B_{ssss} + \varepsilon^2 \left( \frac{3\zeta}{2}-U^2\right ) (B_s^2)_s =0. 
\end{equation}
Remembering that $a:=B_s$ is the relative deviation of the cross-secitonal area,  we can write \eqref{app3} as 
\begin{equation} 
 (\zeta-U^2) a_{ss} -2Ua_{ts} - a_{tt}+ \varepsilon^2 [( \sqrt{\zeta}\pm U )^2 -1]a_{ssss} + \varepsilon^2 \left( \frac{3\zeta}{2}-U^2\right ) (a^2)_{ss }=0. 
 \label{app4} 
\end{equation}
In order to write \eqref{app4} in the canonical form \eqref{BE}, we make the following change of variables: 
\begin{align}
&x=\frac{s-Ut}{\sqrt{\zeta}} , \qquad \tau= t  \\
&\partial_t= \partial_\tau -\frac{U}{\sqrt{\zeta}}\partial_x, \qquad \partial_s =\frac{1}{\sqrt{\zeta}} \partial_x \, . 
\end{align}
In the new variables, \eqref{app4} transforms to 
\begin{equation} 
\label{app5} 
  a_{xx} - a_{\tau \tau}+ \varepsilon^2 \frac{1}{\zeta^2}[( \sqrt{\zeta}\pm U )^2 -1]a_{xxxx} + \varepsilon^2 \frac{1}{\zeta}\left( \frac{3\zeta}{2}-U^2\right ) (a^2)_{xx }=0, 
\end{equation}
which is of the form 
\begin{equation} \label{BE1_augmented}
u_{tt}=u_{xx}+3\alpha(u^2)_{xx} \pm \beta^2 u_{xxxx} \, . 
\end{equation}
The $+$ sign in front of $\beta^2$ term corresponds to the classical Boussinesq equation \eqref{BE1}, and the $-$ sign corresponds to the alternative sign Boussinesq equation obtained from \eqref{BE1} by the substitution $x \rightarrow i x$ and $t \rightarrow i t$, which is also integrable. 
  The constants $\alpha$ and $\beta$ from \eqref{BE1_augmented} are given by 
\begin{equation} 
\label{alpha_beta_eq} 
\alpha \equiv \varepsilon^2 \frac{1}{\zeta}\left( \frac{3\zeta}{2}-U^2\right ) \, , \quad 
\beta_{\pm}^2 \equiv \varepsilon^2 \frac{1}{\zeta^2} \left|  (\sqrt{\zeta}\pm U )^2 -1 \right|   \, . 
\end{equation} 
Note that including $\alpha=0$ makes the equation \eqref{BE1} a linear equation. While this reduction is interesting, to describe both left- and right-running waves at the same time, we believe that equation \eqref{app2} is more universal since it does not contain the $\pm$ signs.

\subsection{Reduction to Korteweg-deVries (KdV) equation}
It is also interesting to investigate the appropriateness of the use of Korteweg-deVries (KdV) approximation for this problem. While the universal Boussinesq-like equation \eqref{app2} is valid for both left and right-travelling waves, its integrable counterpart  \eqref{app5}  is written separately for the left- and right-running waves. However, each version of equation \eqref{app5} itself has both the left- and the right-running waves. Thus, considering \emph{both} left- and right version of  \eqref{app5} generates too much spurious behavior compared to the full system. 
In order to obtain an adequate integrable model for those waves, we reduce \eqref{eq_B} to the KdV equation for waves running in each direction as follows. The KdV approximation for the fluid flow with elastic walls has been considered in the literature before, especially in the context of solitary wave propagation and variable properties of the arteries \cite{De2007,De2009} and reflection of the waves from the branching points \cite{DuWaWe1997}. Thus, we believe it is important to consider the appropriateness of  KdV  approximation in our case as well. 

We are going to proceed as follows. Suppose that $B$ satisfies a KdV equation of the form
\begin{equation} \label{KdV1}
B_t = - c B_s + \varepsilon^2 p B_{sss} + \varepsilon^2 q \left( \frac{B_s^2}{2}\right) 
\end{equation}
Here $p,q,c$ are yet undetermined constants. A subindex $\pm$ will denote right/left running wave, and we shall enforce $c=c_{\pm}=U\pm \sqrt{\zeta}$ as defined by \eqref{t_deriv_approx}. We will drop the $\pm$ sign in the derivation below for brevity. From \eqref{KdV1} we necessarily have $B_t=-cB_s + \mathcal{O}(\varepsilon^2)$.  Differentiating \eqref{KdV1} with respect to $t$ gives 
\begin{equation} 
\label{KdV2}
B_{tt} = - c B_{st} + \varepsilon^2 p B_{ssst} + \varepsilon^2 q \left( \frac{B_s^2}{2}\right)_{t}= - c B_{st} - \varepsilon^2 p c B_{ssss} - \varepsilon^2 q c \left( \frac{B_s^2}{2}\right)_{s} \end{equation}
Next, differentiating \eqref{KdV1} with respect to $s$ gives 
\begin{equation} \label{KdV3}
B_{ts} = - c B_{ss} + \varepsilon^2 p B_{ssss} + \varepsilon^2 q \left( \frac{B_s^2}{2}\right)_{s} 
\end{equation}
From \eqref{KdV2}, \eqref{KdV3} we arrive to 

\begin{equation} 
\label{KdV4}
\begin{aligned} 
& B_{tt} - c^2 B_{ss} + \varepsilon^2 2pc B_{ssss} + \varepsilon^2 q c\left( B_s^2\right)_{s} =0 
\\
& 2UB_{ts} + 2Uc B_{ss} - \varepsilon^2 2U p B_{ssss} - \varepsilon^2 U q \left( B_s^2 \right)_{s} =0 
\end{aligned} 
\end{equation}
Finally, adding equations in \eqref{KdV4} we obtain 
\begin{equation} \label{KdV6}
 (c^2- 2Uc) B_{ss} - 2 U B _{ts} - B_{tt} - \varepsilon^2 2(c- U) p B_{ssss} - \varepsilon^2 (c-U) q \left( B_s^2 \right)_{s} =0 \end{equation}
which matches \eqref{app3}. Indeed, we take $c^2- 2Uc= \zeta - U^2  $,  or $c=c_{\pm}=U \pm \sqrt{\zeta}$, as expected for the speeds of the left- and right-running waves \eqref{t_deriv_approx}.  For the parameters $p$ and $q$, we obtain 
\begin{equation} 
\begin{aligned} 
 - 2p(c- U) & =c^2 -1, \qquad p=-\frac{c^2-1}{2(c-U)}\, , \qquad p_{\pm} = -\frac{c_{\pm}^2-1}{\pm 2\sqrt{\zeta}} ,
\\
 - (c-U) q  & = \frac{3\zeta}{2}-U^2, \qquad q=- \frac{ \frac{3\zeta}{2}-U^2}{c-U}\, , 
\qquad q_{\pm} = - \frac{ \frac{3\zeta}{2}-U^2}{\pm \sqrt{\zeta}}  . 
\end{aligned} 
\end{equation} 
  The equation for $B$ is then 
 \begin{equation} \label{KdV}
B_t + c B_s + \varepsilon^2 \frac{c^2-1}{2(c-U)}  B_{sss} + \varepsilon^2 \frac{ \frac{3\zeta}{2}-U^2}{c-U}   \frac{B_s^2}{2}=0, \qquad 
c=c_{\pm}=U \pm \sqrt{\zeta}\, , 
   \end{equation}
 where we have implicitly assumed that $c$ has the label of $c_{\pm}$ for the left ($-$) and right ($+$) running waves. 
 The KdV equation for $a=B_s$ is
 \begin{equation} \label{KdV_a}
a_t + c a_s + \varepsilon^2 \frac{c^2-1}{2(c-U)}  a_{sss} + \varepsilon^2 \frac{ \frac{3\zeta}{2}-U^2}{c-U} \left( \frac{a^2}{2}\right)_s =0 \, , 
\qquad 
c=c_{\pm}=U \pm \sqrt{\zeta}\, . 
\end{equation}
In order to compare the accuracy of the equation \eqref{KdV}, we perform the simulation with the same initial conditions and parameters as in Figure~\ref{fig:simulations}, separating the solutions into the left- and right-running waves, \emph{i.e.}, taking the plus and minus signs in the corresponding equation \eqref{KdV}. The results for $a=B_s$ are presented on Figure~\ref{fig:KdV}. Note that the speed of propagation is computed quite well as compared to the left panel of Figure~\ref{fig:simulations}, \emph{i.e.}, solution of the full equation \eqref{eq_B}. However, the details of the evolution do not quite match, especially because of the contribution of short wavelength ripples in the KdV equation, which is supposed to be the long-wavelength approximation only. Still, we believe that it is useful to consider the KdV for the  propagation of the wave pattern in either direction.  
\begin{figure}[h]
\centering
\includegraphics[width=0.48\textwidth]{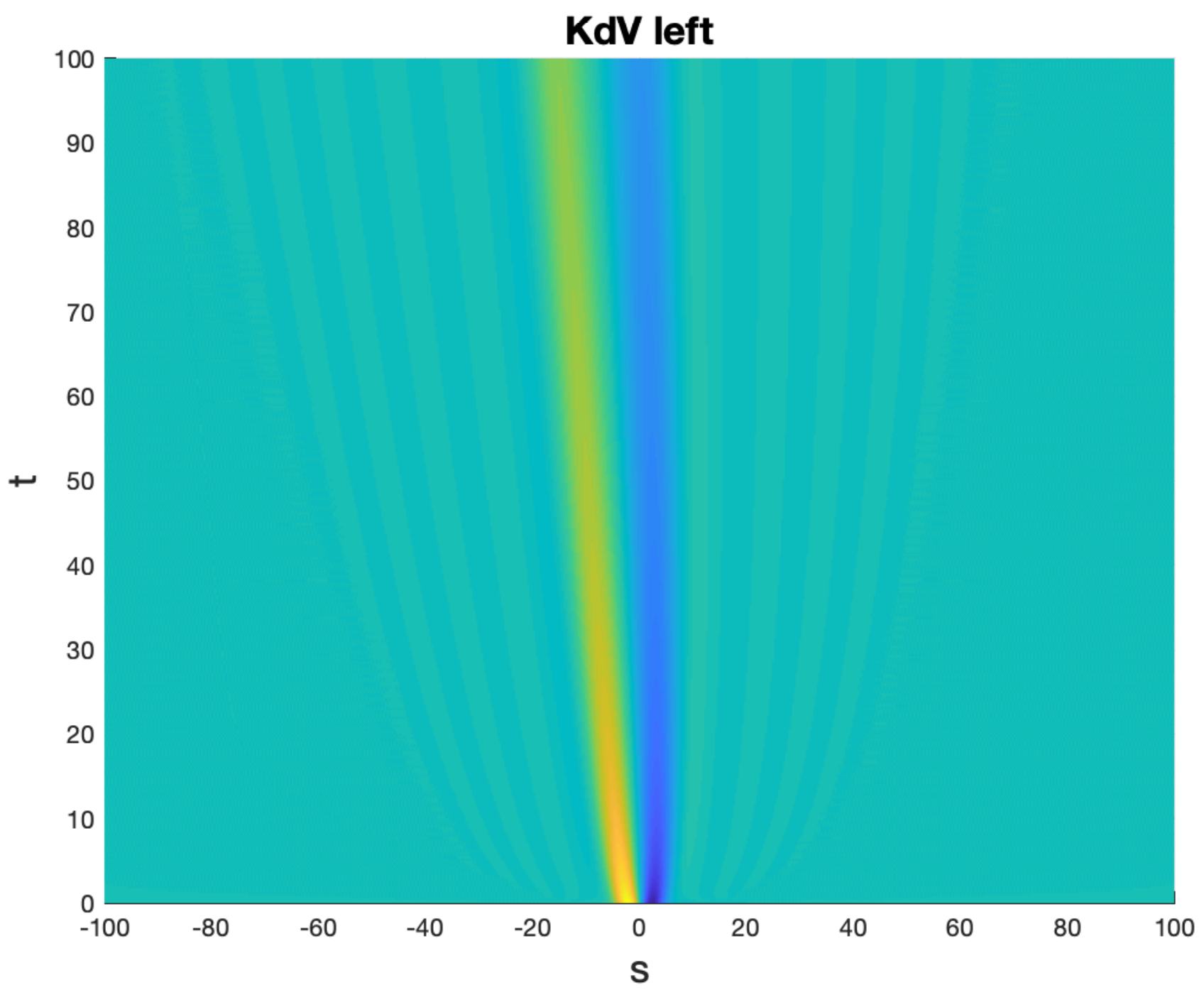}
\includegraphics[width=0.48\textwidth]{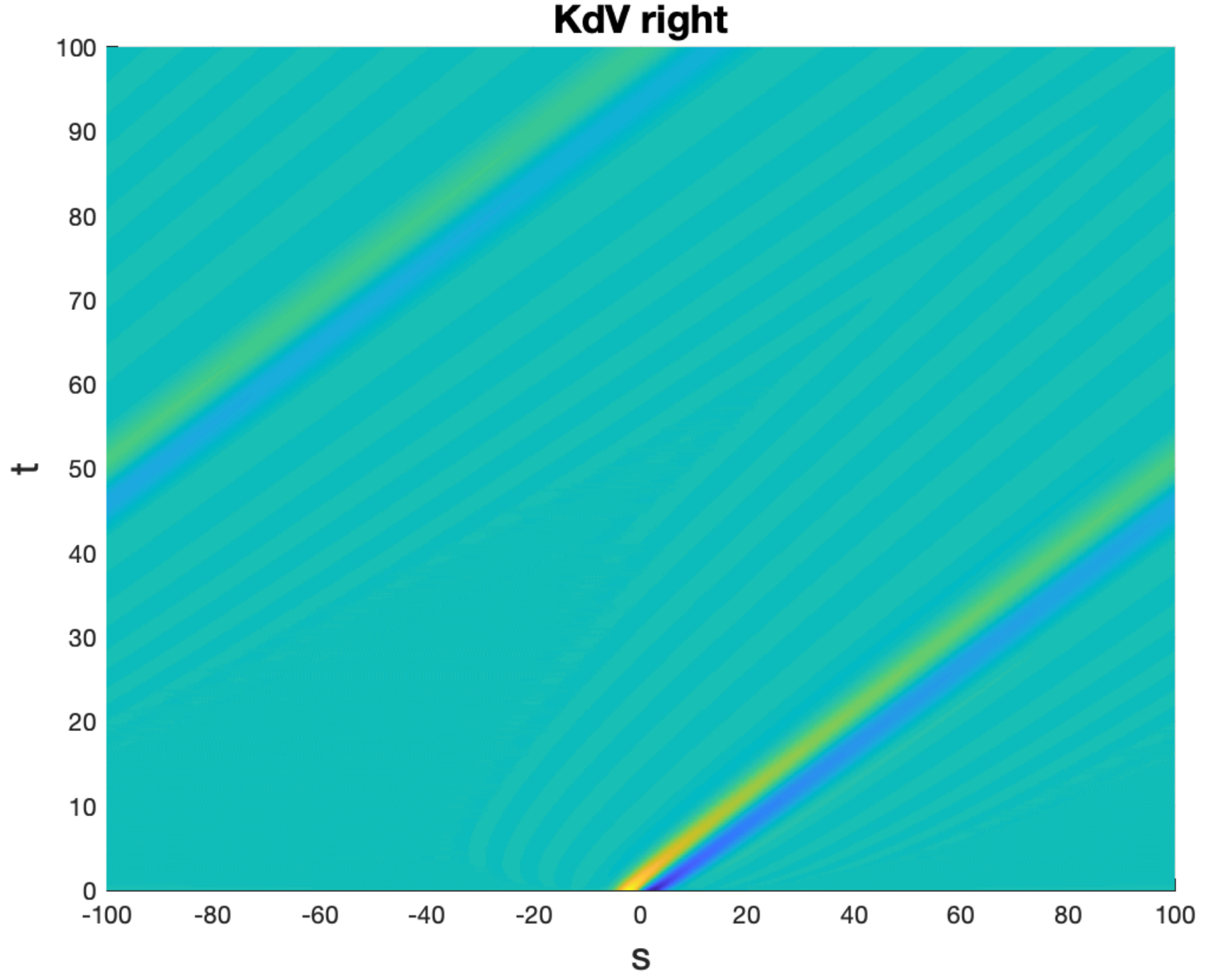}
\caption{Solutions for $a=B_s$ the left-running and right-running KdV equations \eqref{KdV} $a(s,t)$. Left panel: Left-running KdV equation,  corresponding to the $-$ sign in \eqref{KdV}. Right panel: Right-running KdV equation corresponding to the $+$ sign in \eqref{KdV}.   }
\label{fig:KdV}
\end{figure}
\section{Conclusions} 
In this paper, we have derived several novel results for the description of fluid with vorticity flowing in tubes with expandable walls. As our approach is variational in nature, it is applicable to the flow of essentially inviscid fluids. The friction terms can be introduced later in the model by using the Lagrange-d'Alembert's principle of external forces. We have intentionally avoided discussion of the friction terms here, as the nature of fluid friction in blood is very complex and is certainly beyond the scope of this paper. That being said, in order to  apply our theory successfully to the blood flow, the friction between the fluid and the walls will need to be introduced. This is an important feature of the practical applicability of the theory developed here and should be considered later. 

On a more theoretical side, in this manuscript we have concentrated on the purely Lagrangian description of the dynamics. An equivalent Hamiltonian formulation can be derived by constructing an appropriate Poisson bracket. Since the description through the back-to-labels map comes as the Euler-Lagrange equation, the Poisson bracket for the back-to-label description will be canonical. We have not touched upon the Hamiltonian description here to keep the paper concise; the Hamiltonian description will be considered in the follow-up paper. 

Another interesting direction of study, having both a theoretical and a practical value, is a more detailed incorporation of the dynamics of the media outside the tube. In our paper, we have modelled this media as simply exerting a constant pressure, but having no dynamics of its own. This may be a very good approximation for tubes conveying fluid when the outside media is an ideal gas kept at constant pressure. However, in reality, the dynamics of tubes submerged in a resisting media will cause a non-trivial response from the media beyond constant pressure, and thus affect the dynamics in a different ways. This response could be treated in the variational framework if, for example, the outside media is modelled either as an elastic body, or an ideal fluid. This interesting and challenging direction of study will be also considered in the future work. 

\section{Acknowledgements} 

We acknowledge inspiration from fruitful and productive discussions with D. D. Holm, F. Gay-Balmaz and T. S. Ratiu. 
The research of VP was partially supported by the NSERC Discovery Grant program and the University of Alberta. This work has also been made possible by the awarding of a James M Flaherty Visiting Professorship from the Ireland Canada University Foundation, with the assistance of the Government of Canada/avec l'appui du gouvernement du Canada.

\bibliographystyle{unsrt}
\bibliography{Garden-hoses}
\end{document}